%% file: main.tex
\pdfoutput=1

\documentclass{article}

\usepackage{natbib}
\usepackage{Akbar}
\usepackage{nicefrac}

\usepackage{microtype}
\usepackage{graphicx}
\usepackage{subfigure}
\usepackage{booktabs} % for professional tables

\usepackage{hyperref}

\usepackage[accepted]{icml2024}

\usepackage{amsmath}
\usepackage{amssymb}
\usepackage{mathtools}
\usepackage{amsthm}

\usepackage[capitalize,noabbrev]{cleveref}

\theoremstyle{plain}
\newtheorem{theorem}{Theorem}[section]

\newtheorem{lemma}[theorem]{Lemma}
\newtheorem{corollary}[theorem]{Corollary}
\newtheorem{claim}[theorem]{Claim}
\theoremstyle{definition}

\newtheorem{example}[theorem]{Example}
\newtheorem{assumption}[theorem]{Assumption}
\theoremstyle{remark}
\newtheorem{remark}[theorem]{Remark}

\usepackage[textsize=tiny]{todonotes}

\icmltitlerunning{Decomposable Submodular Maximization in Federated Setting}

\begin{document}

\twocolumn[
\icmltitle{Decomposable Submodular Maximization in Federated Setting}



\icmlsetsymbol{equal}{*}

\begin{icmlauthorlist}
\icmlauthor{Akbar Rafiey}{yyy}
\end{icmlauthorlist}

\icmlaffiliation{yyy}{Halıcıoğlu Data Science Institute, University of California, San Diego, USA}
\icmlcorrespondingauthor{Akbar Rafiey}{arafiey@ucsd.edu}

\icmlkeywords{Federated Optimization, Constrained Optimization, Submodular Function}

\vskip 0.3in
]



\printAffiliationsAndNotice{} % otherwise use the standard text.

\begin{abstract}
Submodular functions, as well as the  sub-class of decomposable submodular functions, and their optimization appear in a wide range of applications in machine learning, recommendation systems, and welfare maximization. 
However, optimization of decomposable submodular functions with millions of component functions is computationally prohibitive. Furthermore, the component functions may be private (they might represent user preference function, for example) and cannot be widely shared. To address these issues, we propose a {\em federated optimization} setting for decomposable submodular optimization.  In this setting, clients have their own preference functions, and a weighted sum of these preferences needs to be maximized.  We implement the popular {\em continuous greedy} algorithm in this setting where clients take parallel small local steps towards the local solution  and then the local changes are aggregated at a central server. To address the large number of clients, the aggregation is performed only on a subsampled set. Further, the aggregation is performed only intermittently between stretches of parallel local steps, which reduces communication cost significantly. We show that our federated algorithm is guaranteed to provide a good approximate solution, even in the presence of above cost-cutting measures. Finally, we show how the federated setting can be incorporated in solving fundamental discrete submodular optimization problems such as Maximum Coverage and Facility Location.
\end{abstract}

\input{Introduction}

\input{Preliminaries.tex}
\input{Method.tex}
\input{PracticalFCG.tex}

\input{Discrete_examples.tex}

\input{Conclusion}

\section*{Acknowledgements}
I would like to thank Arya Mazumdar and Barna Saha for their many useful discussions and for reading several drafts of this paper. I would also like to thank Yuichi Yoshida for his valuable discussions and inspiring comments on the paper. Additionally, I am grateful to Nazanin Mehrasa, Wei-Ning Chen, Heng Zhu, and Quanquan C. Liu for insightful comments on this work.

\section*{Impact Statement}

This paper presents work whose goal is to advance the field of 
Machine Learning. There are many potential societal consequences 
of our work, none which we feel must be specifically highlighted here.


\bibliography{references}
\bibliographystyle{icml2024}

%%%%%%%%%%%%%%%%%%%%%%%%%%%%%%%%%%%%%%%%%%%%%%%%%%%%%%%%%%%%%%%%%%%%%%%%%%%%%%%
%%%%%%%%%%%%%%%%%%%%%%%%%%%%%%%%%%%%%%%%%%%%%%%%%%%%%%%%%%%%%%%%%%%%%%%%%%%%%%%
% APPENDIX
%%%%%%%%%%%%%%%%%%%%%%%%%%%%%%%%%%%%%%%%%%%%%%%%%%%%%%%%%%%%%%%%%%%%%%%%%%%%%%%
%%%%%%%%%%%%%%%%%%%%%%%%%%%%%%%%%%%%%%%%%%%%%%%%%%%%%%%%%%%%%%%%%%%%%%%%%%%%%%%
\newpage
\appendix
\onecolumn

 \input{Appendix-Full-participation}

%%%%%%%%%%%%%%%%%%%%%%%%%%%%%%%%%%%%%%%%%%%%%%%%%%%%%%%%%%%%%%%%%%%%%%%%%%%%%%%
%%%%%%%%%%%%%%%%%%%%%%%%%%%%%%%%%%%%%%%%%%%%%%%%%%%%%%%%%%%%%%%%%%%%%%%%%%%%%%%
 \input{Alg1_Omitted_Proofs.tex}

 \input{Alg2_Omitted_Proofs.tex}
 \input{Appendix-discrete-examples}

%%%%%%%%%%%%%%%%%%%%%%%%%%%%%%%%%%%%%%%%%%%%%%%%%%%%%%%%%%%%%%%%%%%%%%%%%%%%%%%
%%%%%%%%%%%%%%%%%%%%%%%%%%%%%%%%%%%%%%%%%%%%%%%%%%%%%%%%%%%%%%%%%%%%%%%%%%%%%%%

\end{document}

%% file: Introduction.tex
\section{Introduction}
%% Submodular Functions
%Submodular optimization is
Submodularity of a set function implies a natural diminishing returns property where the marginal benefit of any given element decreases as we select
more and more elements. Formally, a set function $F\colon 2^E \to \mathbb{R}$ is \emph{submodular} if for any $S \subseteq T \subseteq E$ and $e \in E\setminus T$ it holds that
$
    F(S \cup \{e \}) - F(S) \geq F(T \cup \{e\}) - F(T).
$ Decomposable submodular functions is an important subclass of submodular functions which can be written as sums of several component submodular functions: $F(S) = \sum_{i=1}^N f_i(S)$, for all $S\subseteq E$,
% \begin{align}
%     F(S) = \sum_{i=1}^N f_i(S); \quad \forall S\subseteq E
% \end{align}
where each $f_i:2^E\to \zR$ is a submodular function on the ground set $E$ with $|E|=n$. 

Decomposable submodular functions include some of the most fundamental and well-studied submodular functions such as max coverage, graph cuts, welfare maximization etc., and have found numerous applications in machine learning \citep*{DueckF07,GomesK10, MirzasoleimanBK16,MirzasoleimanKS16,z17}, economics \cite{DobzinskiS06,Feige06,FeigeV06,PapadimitriouSS08,Vondrak08}, and data summarization and recommender systems~\cite{DueckF07,GomesK10,TschiatschekIWB14,lin2011class}. The main approach in these cases is the \emph{centralized} and \emph{sequential} greedy.

%As a consequence of these wide-prevailing applications, maximizing decomposable submodular functions subject to different constraints has attracted significant theoretical/practical interest. The main approach in these cases is the \emph{centralized} and \emph{sequential} greedy approach.  

The need for scalable and efficient optimization methods, which do not require collecting raw data in a central server and ensure secure information collection, is widespread in applications handling sensitive data such as medical data, web search queries, salary data, and social networks. In many such cases, individuals and companies are reluctant to share their data and collecting their data in a central server is a violation of their privacy. Moreover, collecting and storing all data on a single server or cluster is computationally expensive and infeasible for large-scale datasets, particularly when working with high-dimensional data or complex models. Thus, there is a widespread demand for scalable optimization algorithms that are both \emph{decentralized} and prioritize \emph{privacy}. Below are some examples that motivate the focus of this paper.

%a variety of settings and has emerged as a very beneficial property in many combinatorial optimization problems arising in machine learning, graph theory, economics, game theory, to name a few. A novel class of submodular functions are decomposable submodular functions. These are functions that can be written as sums of several ``simple'' submodular functions, i.e.,
%\begin{align}
%    F(S) = \sum_{i=1}^N f_i(S); \quad \forall S\subseteq E
%\end{align}
%where each $f_i:2^E\to \zR$ is a submodular function on the ground set $E$ with $|E|=n$. 

%Decomposable submodular functions encompass many of the examples of submodular functions
%studied in the context of machine learning as well as economics. For example, they are extensively used in economics in the problem of welfare maximization in combinatorial auctions~\cite{DobzinskiS06,Feige06,FeigeV06,PapadimitriouSS08,Vondrak08}.

%% Example

\begin{example}[Welfare Maximization]
The welfare maximization problem aims to maximize the overall utility or welfare of a group of individuals or agents $a_1,\dots, a_N$. In this problem, there is a set of items or goods $E$, and each individual $a_i$ has a certain preference or utility expressed as a submodular function $f_i:2^E\to\zR_{+}$ that assigns a value to each combination of items. The goal is to partition $E$ into disjoint subsets $S_1,\dots,S_N$ in order to maximize the social welfare $\sum_{i=1}^N f_i(S_i)$. In an important special case, called Combinatorial Public Projects \cite{GuptaLMRT10,PapadimitriouSS08}, the goal is to find a subset $S\subseteq E$ of size at most $k$ maximizing $F(S) = \sum_{i=1}^N f_i(S)$. This problem appears in different fields, such as resource allocation, public goods provision, market design, and has been intensively studied \cite{KhotLMM08,LehmannLN06,MirrokniSV08}. An optimal approximation algorithm is known in \emph{value oracle} model in which it is required to have access to the value of $f_i(S)$ for each agent and any $S\subseteq E$ \citep{CalinescuCPV11}. However, in scenarios where agents are hesitant to disclose their data to a central server and storing all data on a single server is computationally infeasible, the demand for a decentralized and private submodular maximization algorithm becomes imperative.
\end{example}

\begin{example}[Feature Selection]
Enabling privacy-protected data sharing among clinical centers is crucial for global collaborations. Consider geographically dispersed hospitals $H_1,\dots,H_N$, with each hospital maintaining its own data that it is unwilling to share. The goal is to identify a small subset of features that effectively classifies the target variable across the entire dataset over all hospitals. A decentralized and efficient feature selection algorithm is crucial for uncovering hidden patterns while maintaining data ownership. By adopting a decentralized approach, hospitals can balance collaborative knowledge discovery and data privacy. One approach is to maximize a submodular function capturing the mutual information between features and the class labels \cite{krause2005near}.  
\end{example}

\noindent{\bf Federated setting for learning and optimization.}
%% Motivation for Federated Setting
% The need for scalable and efficient optimization methods, which do not require collecting raw data in a central server and ensure secure information collection, is widespread in applications handling sensitive data such as medical data, web search queries, salary data, and social networks. In many such cases, individuals and companies are reluctant to share their data and collecting their data in a central server is a violation of their privacy. Moreover, collecting and storing all data on a single server or cluster is computationally expensive and infeasible for large-scale datasets, particularly when working with high-dimensional data or complex models.
\emph{Federated setting} \cite{konevcny2016federated,McMahanMRHA17} is a novel and practical framework that addresses issues regarding privacy, data sharing, and centralized computation. On one hand, it is a distributed and collaborative approach that allows multiple parties, such as different organizations or devices, to train a shared model or collaboratively optimize an objective function while keeping their data locally. This approach helps to protect the privacy of data by ensuring that the raw data is never shared or moved outside of the individuals' systems. Instead, only the model updates are exchanged and aggregated to improve the shared model and improve the objective value. 

%Federated learning is a recently popular framework that enables training a model on data distributed across multiple clients without collecting all the data in a central location. Instead, the model is trained on the individual machines and the updates are then aggregated to improve the global model. Standard optimization methods (stochastic gradient descent in particular) have been adapted to the federated setting with much success.

On the other hand, the federated framework reduces the amount of data that needs to be transferred and processed at any one time, which can significantly reduce the computational complexity of the overall process. Additionally, federated setting can also take advantage of the computational resources available at each party, such as the processing power of mobile devices or edge devices, which can further reduce computational load on the server. This way, federated learning can train models more efficiently, even with large-scale datasets and complex models, and provide a scalable solution for distributed learning.

\noindent{\bf Problem definition and setting.} 
In this paper, we introduce the problem of maximizing a submodular function in the federated setting. % \texttt{Federated Submodular Maximization}. 
Let $E$ be a ground set of size $n$ and $c_1,\dots,c_N$ be $N$ clients each of whom has a \emph{private} interest over $E$. Each client's interest is expressed as a submodular function. Let $f_i:2^E\to \mathbb{R}_{+}$ be the associated submodular function of the $i$-th client. A central server wants to solve the following constrained distributed optimization model
\begin{align}
\label{eq:objective-constrained}
    \max_{S\in \mc{I}}\left\{F(S)=\sum_{i=1}^N p_i f_i(S)\right\},
\end{align}
where $\mc{I}$ is the independent sets of a matroid $\mc{M}$ with ground set $E$, and
%\begin{align}
%\label{eq:objective-unconstrained}
%    \max_{S\subseteq E}\left\{F(S)=\sum_{i=1}^N p_i f_i(S)\right\}
%\end{align}
 $p_i$ are pre-defined weights such that $\sum_{i=1}^N p_i=1$. For instance, they can be set to $1/N$, or the fraction of data owned by each client.
% Then the objective is 
% \begin{align}
% \label{eq:objective-unconstrained}
%     \argmax_{S\subseteq E} F(S).
% \end{align}
The constraint implies sets of particular properties, e.g., subsets of size at most $k$. % or subsets that belong to a matroid $\mc{M}$, the optimization problem is formulated as
Note that, the unconstrained optimization is a special case of this.
%The focus of this work is on the constrained case \eqref{eq:objective-constrained} which is more general than the unconstrained one \eqref{eq:objective-unconstrained}.

%\paragraph{Three characteristics.} 
In the optimization problem  \eqref{eq:objective-constrained}
the  data can be massively distributed over 
the number of clients  $N$, which can be huge. Moreover unlike the traditional distributed setting, in the federated setting the server does not have control over clients' devices nor on how data is distributed. For example, when a mobile phone is turned off or WiFi access is unavailable, the central server will lose connection to this device. Furthermore, client's  objective can be very different depending on their local datasets. To minimize communication overhead and server computation load, the number of communication rounds need to be minimal. 

\noindent{\bf Constraints.} 
We formally discuss factors of efficiency and restrictions that should be considered.
%\begin{itemize}[leftmargin=*]
 %\item[1.] \underline{Privacy}: 
 
 \noindent{1. \underline{Privacy}:} One of the main appeals of decentralized and federated setting is to preserve the privacy. There are several models of privacy and security that have been considered in the literature such as Differential Privacy (DP) and Secure Aggregator (SecAgg), and a mix of these two. Single-server SecAgg is a cryptographic secure multi-party computation (MPC) that enables clients to submit vector inputs, such that the server (an aggregator) can \emph{only} decipher the combined update, not individual updates. This is usually achieved via additive masking over a finite group \cite{BellBGL020-SecAgg,BonawitzIKMMPRS16-SecAgg}. Note that secure aggregation alone does not provide any privacy guarantees. To achieve a DP-type guarantee, noise can be added locally, with the server aggregating the perturbed local information via SecAgg. This user-level DP framework has recently been adopted in private federated learning \cite{AgarwalSYKM18,AgarwalKL21,KairouzL021,DBLP:conf/icml/ChenCKS22, DPFL-wang2023federated}.
 
 In this paper we use Single-server SecAgg model of privacy, a dominant and well-established approach in the field. We leave other notions of privacy, such as a mix of DP and SecAgg, for future works. There has been a recent and concurrent progress towards this direction in terms of cardinality constraints \cite{DPFL-wang2023federated}.
 
  %\item[2.] \underline{Communication and bit complexity}:
 \noindent{2. \underline{Communication and bit complexity}:}
   There are a few aspects to this, firstly the number of communication rounds should be as small as possible. Second, the information communicated between should require low bandwidth and they better require small bit complexity to encode.
   
    %\item[3.] \underline{Convergence and utility}: 
    \noindent{3. \underline{Convergence and utility}:}
    While the above impose strong restrictions, a good decentralized submodular maximization algorithm should not scarify the convergence rate by too much and should yield to an accurate and acceptable result in comparison to the centeralized methods.
%\end{itemize}

\noindent{\bf Our contributions.} We present the first federated (constrained) submodular maximization algorithm converging close to optimum guarantees known in centralized settings.

$\bullet$ We propose a decentralized version of the popular Continuous Greedy algorithm 
\texttt{Federated Continuous Greedy} (\texttt{\textsc{Fed}CG)} and prove its convergence whenever the  client functions are nonnegative monotone submodular achieving the optimal multiplicative approximation factor $(1-\nicefrac{1}{e})$ with a small additive (\Cref{sec:FCG}). 

$\bullet$ We incorporate  important and practical scenarios that are relevant for federated setting such as partial client selection, low communication rounds and computation cost. We give rigorous theoretical guarantees under each scenarios matching the optimal multiplicative approximation of $(1-\nicefrac{1}{e})$ and small additive error (\Cref{sec:fcg+}). 

$\bullet$ We introduce a new algorithm that serves as a discrete federated optimization algorithm for submodular maximization. Its convergence and applications to discrete problems such as \texttt{Facility Location} and \texttt{Maximum Coverage} are explored (Section \ref{sec:discrete-fl}).

% scaling linearly with the rank of the matroid and a gradient dissimilarity factor as well as an error caused by clients' local steps (Section \ref{sec:fcg+}).

%$\bullet$ We complement our theoretical results with an empirical study of the performance of our algorithms. (Due to space limitation, in the Appendix, we introduce a new discrete federated optimization algorithm and explore its convergence for problems such as \texttt{Facility Location} and \texttt{Max Coverage}). 
\subsection{Related Work}

\paragraph{FedAvg, its convergence, and assumptions.} The concept of Federated Learning (FL) \cite{McMahanMRHA17} has found application in various domains such as natural language processing, computer vision, and healthcare. The popular FL algorithm, Federated Averaging (\texttt{\textsc{Fed}Avg}), is an extension of Local SGD that aims to reduce communication costs in distributed settings \cite{GorbunovHR21-localsgd,Stich19,WangJ21-cooperativeSGD,YuYZ19-FEDAVG-heterogeneous}. However, despite its practical benefits in addressing privacy, data heterogeneity, and computational constraints, it may not converge to a ``good enough'' solution in general \cite{PathakW20,Zhang2020}. Analyzing the convergence of \texttt{\textsc{Fed}Avg} and providing theoretical guarantees is challenging and necessitates making certain assumptions. Assumptions related to bounded gradients, convexity, Lipschitzness, statistical heterogeneity, and bounded variance of stochastic gradients for each client have been explored in recent works \cite{SCAFFOLD-Karimireddy20,LiHYWZ20,WoodworthPSDBMS20,WangTSLMHC19,YuYZ19-FEDAVG-heterogeneous}.

\noindent{\bf Decentralized / distributed submodular maximization.}
The main approach to submodular maximization is the greedy approach which in fact, in the centralized setting yields the tight approximation guarantee in various scenarios and constraints e.g., see \cite{NemhauserWF78,Vondrak08,CalinescuCPV11}. Centeralized submodular maximization under privacy constraints is an active research area \cite{MitrovicB0K17,RafieyY20,ChaturvediNZ21}. However the sequential nature of the greedy approach makes it challenging to scale it to massive datasets. This issue is partially addressed by the means of Map-Reduce style algorithms \cite{KumarMVV15-mapreduc} as well as several elegant algorithms in the distributed setting \cite{MirzasoleimanKS16,BarbosaENW15}. Recent work of \citet{MokhtariHK18-Decentralized} ventures towards decentralized submodular maximization for continuous  submodular functions. In general continuous submodular functions are not convex nor concave and there has been a line of work to optimize continuous submodular functions using SGD methods \cite{HassaniSK17-SGD-submodular}. \citeauthor{MokhtariHK18-Decentralized} under several assumptions, such as assuming clients' local objective functions are monotone, DR-submodular, Lipschitz continuous, and have bounded gradient norms, prove that \emph{Decenteralized Continuous Greedy} algorithm yields a feasible solution with quality $O(1- 1/\mathrm{e})$ times the optimal solution. The setting in \cite{MokhtariHK18-Decentralized} is fundamentally different from the federated setting in a sense that they require sharing gradient information of the clients with the server or with the neighboring nodes in an underlying graph. Perhaps the most closely related method to our work is due to \citet{dadras2022federatedFW,ZhangZGZW22-decentralized}. \citet{dadras2022federatedFW} consider Frank-Wolfe Algorithm \cite{frank1956algorithm} in the federated setting and propose \texttt{Federated Frank-Wolfe} (\texttt{\textsc{Fed}FW}) algorithm and analyze its convergence for both convex and non-convex functions and under the $L$-Lipschitzness and bounded gradients assumptions. 

%Note that Continuous Greedy algorithm and Frank-Wolfe algorithm are very similar, see \cite{BianMB017,MokhtariHK18} for approximation guarantee of Frank-Wolfe algorithm in centralized setting for monotone submodular function subject to bounded convex body constraint.

%% file: Preliminaries.tex
\section{Preliminaries}
Let $E$ denote the ground set, $|E|=n$. For a vector $\mb{x}\in\zR^{|E|}$ and a set $S\subseteq E$, $\mb{x}(S)$ denotes $\sum_{e\in S}\mb{x}(e)$. A submodular function $f:2^E\to \zR$ is monotone if $f(S)\leq f(T)$ for every $S\subseteq T\subseteq E$. Throughout this paper we assume $f(\emptyset)=0$.

\noindent{\bf Multilinear extension.}
    The multilinear extension $\widehat{f}:[0, 1]^{|E|}\to \zR$ of a set function $f:\{0,1\}^{|E|}\to \zR$ is
    \begin{align*}
        \widehat{f}(\bx) &=\sum_{S\subseteq E}f(S)\prod\limits_{e\in S}\bx(e)\prod\limits_{e\not\in S}(1-\bx(e))= \mathbb{E}_{R\sim \mb{x}}[f(R)]
    \end{align*}
    where $R\subseteq E$ is a random set that contains each element $e \in E$  independently with probability $\mb{x}(e)$ and excludes it with probability $1-\mb{x}(e)$.
    We write $R \sim \mb{x}$ to denote that $R\subseteq E$ is a random set sampled according to $\mb{x}$. 
   
    Observe that for all $S\subseteq E$ we have $\widehat{f}(\mb{1}_S) = f(S)$. For monotone non-decreasing submodular function $f$, $\widehat{f}$ has the following properties \cite{CalinescuCPV11} that are crucial in analyses of our algorithms:
    
    % \begin{itemize}[leftmargin=*]
      \noindent{1.} $\widehat{f}$ is monotone, meaning $\frac{\partial \widehat{f}}{\partial \mb{x}(e)}\geq 0$. Hence, $\nabla \widehat{f}(\mb{x})=(\frac{\partial \widehat{f}}{\partial \mb{x}(1)},\dots,\frac{\partial \widehat{f}}{\partial \mb{x}(n)})$ is a nonnegative vector. 
      
      \noindent{2.} $\widehat{f}$ is concave along any direction $\mathbf{d}\geq \mathbf{0}$.
  %\end{itemize}
  
Note that $\frac{\partial \widehat{f}}{\partial \mb{x}(e)}=\mathbb{E}_{R\sim \mb{x}}[f(R\cup\{e\})-f(R\setminus \{e\})]$. That is the expected marginal contribution for $e$ where the expectation is taken over $R\subseteq E\setminus \{e\}$ sampled according to $\mb{x}$. By submodularity, for any $\mb{x}\in [0,1]^n$,
\begin{align}
    |\nabla \widehat{f}(\mb{x})|_{\infty} \leq \max_{e\in E}f(\{e\}):=m_f.
\end{align}
%     The following is well known:
% \begin{proposition}[\citet{CalinescuCPV11}]
% \label{prop:MLE-properties}
%   Let $\widehat{g}\colon {[0,1]}^{|E|} \to \mathbb{R}$ be the multilinear extension of a monotone submodular function $g\colon 2^E \to \mathbb{R}$. Then
%   \begin{itemize}
%       \item [1.] $\widehat{g}$ is monotone, meaning $\frac{\partial \widehat{g}}{\partial \mb{x}(e)}\geq 0$. Denote $E= \{1, 2, \dots, n\}$. Hence, $\nabla \widehat{g}(\mb{x})=(\frac{\partial \widehat{g}}{\partial \mb{x}(1)},\dots,\frac{\partial \widehat{g}}{\partial \mb{x}(n)})$ is a nonnegative vector. 
%       \item[2.] $\widehat{g}$ is concave along any direction $\mathbf{d}\geq \mathbf{0}$.
%   \end{itemize}
% \end{proposition}
%\subsection{Matroids and Matroid Polytopes}
\noindent{\bf Matroids and matroid polytopes.}
 A pair $\mathcal{M} = (E,\mc{I})$ of a set $E$ and $\mc{I}\subseteq 2^E$ is called a \emph{matroid} if 1) $\emptyset \in \mc{I}$, 2) $A\in \mc{I}$ for any $A \subseteq B\in \mc{I}$, and 3) for any $A,B\in \mc{I}$ with $|A| < |B|$, there exists $e \in B\setminus A$ such that $A\cup \{e\}\in \mc{I}$.
% \begin{enumerate}
% \item[(1)] $\emptyset \in \mc{I}$,
% \item[(2)] $A\in \mc{I}$ for any $A \subseteq B\in \mc{I}$, and
% \item[(3)] for any $A,B\in \mc{I}$ with $|A| < |B|$, there exists $e \in B\setminus A$ such that $A\cup \{e\}\in \mc{I}$.
% \end{enumerate}
We call a set in $\mathcal{I}$ an \emph{independent set}. We sometimes abuse notation and use $S\in\mc{M}$. The \emph{rank function} $r_{\mc{M}}\colon 2^E \to \mathbb{Z}_+$ of $\mc{M}$ is 
$r_{\mc{M}}(S)=\max\{|I|:I\subseteq S, I\in \mc{I}\}$.
An independent set $S\in \mc{I}$ is called a \emph{base} if $r_{\mc{M}}(S)=r_{\mc{M}}(E)$. We denote the rank of $\mc{M}$ by $r(\mc{M})$.
The \emph{matroid polytope} $\mc{P}(\mc{M}) \subseteq \zR^E$ of $\mc{M}$ is $\mc{P}(\mc{M})=\mathrm{conv}\{\mathbf{1}_I : I\in \mc{I}\}$ where $\mathrm{conv}$ denotes the convex hull.
Or equivalently~\cite{edmonds2003submodular}, 
$
  \mc{P}(\mc{M})=\left\{\mb{x}\geq \mb{0} : \mb{x}(S)\leq r_{\mc{M}}(S) ;~ \forall S\subseteq E \right\}.
$
% The base polytope is $B(\mc{M}) = \{\mb{y} \in \mc{P}(\mc{M}) \mid \mb{y}(E)=r_{\mc{M}}(E)\}$. 
 
 %The extreme points of $B(\mc{M})$ are the characteristic vectors of the bases of $\mc{M}$.

\noindent{\bf The Continuous Greedy Algorithm.}
Our algorithms for maximizing a submodular function in federated settings are based on the \texttt{Continuous Greedy} (\texttt{CG}) algorithm. We briefly explain this algorithm. The results mentioned are from \cite{CalinescuCPV11,Vondrak08}. Let $\mc{M}=(E,\mc{I})$ be a matroid and $\mc{P}(\mc{M})$ be its matroid polytope of rank $r$, let $f$ be a nonnegative and monotone submodular function and $\widehat{f}$ be its multilinear extension. \texttt{CG} starts with $\mb{x}^{(0)}=\mb{0}$. For every $t\in\{0,1,2,\dots,T-1\}$ it computes $\mb{x}^{(t+1)}$ using the following update step $\mb{x}^{(t+1)}\gets \mb{x}^{(t)} + \eta \mb{v}^{(t)}$, 
where $\mb{v}^{(t)}=\argmax_{\mb{w}\in\mc{P}}\langle \mb{w},\nabla\widehat{f}(\mb{x}^{(t)})\rangle$. For $\texttt{OPT}=\max_{\mb{x}\in\mc{P}}\widehat{f}(\mb{x})$ we have $(1-(1-\eta)^T)\texttt{OPT}\leq  \widehat{f}(\mb{x}^{(T)})+C\eta^2/2$. Here, the constant $C$ depends on the Lipschitz of the function, and $\mb{x}^{(T)}\in\mc{P}$ as it is a convex combination of vectors from the polytope. For $\eta=1/T$ and large enough $T$ we get $(1-1/\mathrm{e})\texttt{OPT}\leq \widehat{f}(\mb{x}^{(T)})+\epsilon$. Given $\mb{x}^{(T)}\in\mc{P}$, there are rounding procedures to obtain $S\in\mc{I}$ with $\widehat{f}(\mb{x}^{(T)})\leq f(S)$. The approximation factor $1-1/\mathrm{e}$ is the best possible assuming $\mb{P}\neq \mb{NP}$ \citep{Feige98-hardness}.

%% file: Method.tex
\section{Federated Continuous Greedy}
\label{sec:FCG}
In this section, we propose our \texttt{Federated Continuous Greedy}~(\texttt{\textsc{Fed}CG)} method. We start with a simplistic scenario of  federated model with full participation which already shows some of the challenges that we have to overcome before delving into the partial participation model which is more computationally feasible. %Note that we want to optimize
   % $$ \max_{S\in \mc{I}}\left\{F(S)=\sum_{i=1}^N p_i f_i(S)\right\}$$

Consider optimization problem~\eqref{eq:objective-constrained} where each $f_i$ is a nonnegative monotone submodular function, $\sum_{i=1}^{N} p_i=1$, and $\mc{I}$ is the independent sets of matroid $\mc{M}=(E,\mc{I})$ of rank $r$.

% \paragraph{Bit complexity and accuracy trade-off.} 

\noindent{\bf Bit complexity and accuracy trade-off.} For every client $i$ and every $\mb{x}\in \mc{P}$, the vector $\mb{v}_i=\argmax_{\mb{w}\in\mc{P}}\langle \mb{w},\nabla\widehat{f_i}(\mb{x})\rangle$ is determined by maximizing a linear function $\langle \mb{w},\nabla\widehat{f_i}(\mb{x})\rangle$ over matroid polytope $\mc{P}$. This problem can be solved very efficiently. We can assume that $\mb{v}_i$ is a vertex of $\mc{P}$ and furthermore, since $\nabla\widehat{f_i}$ is a nonnegative vector, that this vertex corresponds to a base of matroid $\mc{M}$. Hence, without loss of generality $\mb{v}_i$ is the indicator vector of a base with $r$ ones and $n-r$ zeros. Thus it can be encoded using $O(r\log(n))$ bits, sublinear in the size of the ground set. On the other hand, the  vector $\nabla\widehat{f_i}(\mb{x})$ itself requires $\Tilde{O}(n)$ bits for encoding. In what follows we will see how restricting the bit complexity effects the accuracy, and the amount of computation that server should do.

\noindent{\bf \texttt{\textsc{Fed}CG} with full participation.} First consider the case where clients can send their gradients. We proceeds in rounds. Initially $\mb{x}^{(0)}=0$. On the $t$-th round, first, 
the central server broadcasts the latest model $\mb{x}^{(t)}$ to all clients. Each client $i$ after receiving the update sets $\mb{x}_{i}^{(t)}=\mb{x}^{(t)}$ and computes $\nabla\widehat{f}_i(\mb{x}^{(t)})$. The server then aggregates local information via SecAgg, and computes $\nabla\widehat{F}(\mb{x}^{(t)})=\sum_{i=1}^N p_i \nabla\widehat{f}_i({\mb{x}^{(t)}})$. After receiving $\nabla\widehat{F}(\mb{x}^{(t)})$, the server computes $\mb{v}^{(t)}=\argmax_{\mb{w}\in \mc{P}}\langle \mb{w},\nabla\widehat{F}(\mb{x}^{(t)})\rangle$ by maximizing a linear function subject to the matroid constraint and produces the new global model with learning rate $\eta$: $\bx^{(t+1)}\gets \bx^{(t)} + \eta \mb{v}^{(t)}$.
% \begin{align}
%     \bx^{(t+1)}\gets \bx^{(t)} + \eta \mb{v}^{(t)}
% \end{align}
It is clear that, similar to the centeralized \texttt{CG}, large enough $T$ yields $F(\mb{x}^{(T)})\geq (1-\nicefrac{1}{\mathrm{e}})\texttt{OPT}$. This simple framework has an advantage over centralized \texttt{CG}, it is taking advantage of the computational resources available at each client.

Second consider a more challenging case where clients can send at most $\tilde{O}(r)$ bits information. We see how this restriction effects the accuracy. Our algorithm,  
\texttt{\textsc{Fed}CG}, proceeds in rounds. Initially $\mb{x}^{(0)}=0$. On the $t$-th round, first, 
the central server broadcasts the latest model $\mb{x}^{(t)}$ to all clients. Each client $i$ after receiving the update sets $\mb{x}_{i}^{(t)}=\mb{x}^{(t)}$ and performs one step of continuous greedy approach to find a direction that best aligns with her local gradient:
\begin{align}
\label{eq:v}
\mb{v}_i^{(t)}\gets\argmax_{\mb{v}\in \mc{P}} \langle \mb{v},\nabla \widehat{f}_i(\bx_i^{(t)})\rangle
\end{align}
Lastly, clients send their update directions $\mb{v}_1^{(t)},\dots,\mb{v}_N^{(t)}$ to the secure aggregator to compute $\Delta^{(t)}=\sum_{i=1}^N p_i \mb{v}_i^{(t)}$. After computing $\Delta^{(t)}$, the server produces the new global model with learning rate $\eta$:
\begin{align}
    \bx^{(t+1)}\gets \bx^{(t)} + \eta \Delta^{(t)}
\end{align}

Even in this unrealistic setting where all clients participate in each round the convergence analysis requires new insights. In order to provide an approximation guarantee for our algorithm, we shall obtain a lower bound on the function value improvement by taking the direction $\Delta^{(t)}$, that is providing a lower bound for $\langle \Delta^{(t)},\nabla\widehat{F}(\mb{x}^{(t)})\rangle$. However, each $\mb{v}_i^{(t)}$ is the projection of the local gradient into the matroid polytope and it does not carry information about the magnitude of the expected marginal contributions i.e. $\|\nabla\widehat{f}_i(\mb{x}^{(t)})\|$. Without assuming an assumption on the \emph{heterogeneity} of local functions one can construct examples where a single element and corresponding client's marginal contribution are significantly more dominant than others and hence taking direction $\Delta^{(t)}$ results in a very bad approximation guarantee.

%The following lemma analyzes the improvement we gain by using the updates from all clients.

We therefore need an assumption that acts as a tool in constraining the level of heterogeneity which poses a significant obstacle in federated optimization. A common way to handle heterogeneity is to impose a bound on the magnitude of gradients over each clients' local function i.e $\|\nabla \widehat{f}_i(\mb{x})\|\leq \gamma$. This types of assumption is not only common in the literature regarding submodular function maximization in decentralized settings \citep{MokhtariHK18-Decentralized,Zhang2020}, but also in studies of convex and non-convex optimization in federated learning \citep{chen2022optimal,dadras2022federatedFW,SCAFFOLD-Karimireddy20,YuYZ19-FEDAVG-heterogeneous,LiHYWZ20}. 

In the case of submodular functions, each coordinate of the gradient of the multilinear extension corresponds to the marginal gain of adding one single element. In this paper we impose the following assumption which is much more relaxed than assuming a bound on the magnitude of the gradients from each client.

% coordinate of the vector of the multilinear extension corresponds to the marginal gain of adding one single element i.e. $\nicefrac{\partial \widehat{f_i}(\mb{x})}{\partial \mb{x}(e)}=\E_{X\sim \mb{x}}[f_i(X\cup\{e\})-f_i(X)]$. In this paper we impose the following assumption which is much more relaxed than assuming a bound on the magnitude of the gradients from each client.

% It should be noted that the assumption of bounded gradients is not only a natural constraint, but it is also a crucial one for providing convergence guarantees. This is due to the fact that in scenarios where a single element and corresponding client's marginal contribution are significantly more dominant than others, until we select that client and receive her update, any attempted optimization would be entirely ineffective.

% In this paper, rather than imposing a bound on the $L_{+\infty}$ norm of the gradients, we propose a more relaxed assumption.

\begin{assumption}
\label{ass:bounded-gradients}
    For all $i=1,\dots,N$ and $t=1,\dots,T$ we have 
    $| \nabla \widehat{f}_i(\mb{x}^{(t)}) - \nabla\widehat{F}(\mb{x}^{(t)}) |_{\infty} \leq \gamma_t$. Note that, by submodularity and monotonicity we have
    \begin{align}
    \label{eq:bound-gamma}
        \max_{t\in [T]}\gamma_t \leq 2\max_{i}\max_{e\in E}f_i(\{e\}) = 2\max_{i}m_{f_i}.
    \end{align}
\end{assumption}

Monotonicity of $f_i$ implies that for every $\mb{x}\leq \mb{y}$ coordinate-wise, it holds that $\widehat{f}_i(\mb{x})\leq \widehat{f}_i(\mb{x})$. Additionally, gradients are antitone i.e., for every $\mb{x}\leq \mb{y}$ coordinate-wise, it holds that $\nabla\widehat{f}_i(\mb{x})\geq \nabla\widehat{f}_i(\mb{y})$. Thus as the algorithm advances and $t$ grows, $\gamma_t$'s change but never exceed the upper bound in \eqref{eq:bound-gamma}. Additionally, in numerous instances, $\max_{t\in [T]}\gamma_t$ is relatively small. For instance, in \texttt{Max Coverage} problem each $f_i(S)$ is either 0 or 1 , depending if a client is covered by $S$ or not, thus for this problem $\gamma=1$.
% \begin{example}[Max Coverage]
%      Let $\mc{C}=\{C_1,\dots,C_N\}$ be a set of clients and $E=\{G_1,\dots,G_n\}$ be a family of sets where each $G_i\subseteq \mc{C}$ is a group of clients. In the \texttt{Max k-Coverage} problem the objective is to select at most $k$ groups of clients from $E$ such that the maximum number of clients are covered, i.e., the union of the selected groups has maximal size. One can formulate this problem as follows. For every $i\in [N]$ and $A\subseteq [n]$ define $f_i(A)$ as
% \begin{align*}
%     f_i(A)=
%     \begin{cases}
%       1 & \text{if there exists $a\in A$ such that $C_i\in G_a$} ,\\
%       0 & \text{otherwise.}
%    \end{cases}
% \end{align*}
% Each $f_i$ is monotone and submodular and $F(A)=\sum_{i\in [N]}f_i(A)$ is monotone and submodular as well. For this problem $\gamma=1$.
%\end{example}

We now have enough ingredients to prove the following convergence theorem for the case where all clients participate in every communication round. Let $D=\sum_{t=1}^T\gamma_t$
\begin{theorem}[Full participation]
\label{thm:full-participation}
    Let $\mc{M}$ be a matroid of rank $r$ and $\mc{P}$ be its matroid polytope. Under the full participation assumption and Assumption~\ref{ass:bounded-gradients}, for every $\eta > 0 $, Algorithm~\ref{alg:FCG} returns a $\mb{x}^{(T)}\in\mc{P}$ such that 
    \begin{align*}
        &\left( 1- (1-\eta)^T \right)\texttt{OPT} \leq \widehat{F}(\mb{x}^{(T)})  + \eta r\sum_{t=1}^T\gamma_t+\frac{T\eta^2r^2 m_F}{2}
    \end{align*}
    In particular, for large enough $T$,  setting $\eta = 1/T$, Algorithm~\ref{alg:FCG} requires at most $\Tilde{O}(r)$ bits of communication per user per round ($\Tilde{O}(NTr)$ in total) and obtains
    \[
         (1- 1/\mathrm{e})\texttt{OPT} \leq \widehat{F}(\mb{x}^{(T)}) + \left(\frac{rD}{T}+\frac{r^2 m_F}{2T}\right).
     \]
\end{theorem}

\begin{algorithm}[tb]
  \caption{Federated Continuous Greedy (\texttt{\textsc{Fed}CG)}}
  \label{alg:FCG}
  \begin{algorithmic}[1]
    \STATE {\bfseries Input: }{Matroid polytope $\mc{P}$, number of communication rounds $T$, learning rate $\eta$, and $K$.}
    \STATE $\bx^{(0)}=\mb{0}$
    \FOR{$t=0$ to $T-1$}
        \STATE Server selects a subset of $K$ \emph{active} clients $A^{(t)}$ according to Client Sampling Scheme, and sends $\bx^{(t)}$ to them.
        \FOR{Client $i$ in $A^{(t)}$ in parallel }
            %\STATE $\bx_i^{(t)}\gets \bx^{(t)}$
            \STATE $\mb{v}_i^{(t)}\gets\argmax_{\mb{v}\in \mc{P}} \langle \mb{v},\nabla \widehat{f}_i(\bx^{(t)})\rangle$ 
            \STATE Send $ \mb{v}_i^{(t)}$ back to the secure aggregator.
        \ENDFOR
        \STATE SecAgg: $\Delta^{(t)}=\frac{1}{|A^{(t)}|}\sum_{i\in A^{(t)}} \mb{v}_i^{(t)}$
        \STATE Server updates: $\bx^{(t+1)}\gets \bx^{(t)} + \eta\Delta^{(t)}$
    \ENDFOR
    %\STATE {\bfseries Output: } 
    \STATE Apply a proper rounding scheme on $\bx^{(T)}$ to obtain a solution for \eqref{eq:objective-constrained}
  \end{algorithmic}
\end{algorithm}
%\end{minipage}
\subsection{Partial Participation and Client Selection}
Client sampling in the FL optimization framework is imperative for various practical reasons, including the following:

% \begin{itemize}
% [leftmargin=*]
    $\bullet$ Large scale and dynamic nature. In real-world applications, a server usually serves several billions of devices/clients who can join or leave the federated optimization system due to several reasons like intermittent connectivity, technical issue, or simply based on their availability or preferences. Hence, it is 
    computationally inefficient and often impossible to get updates from all clients.

  $\bullet$ Communication and bandwidth. On one hand, waiting for the slowest client to finish can increase the expected round duration as the number of participating clients per training increases, a phenomenon known as ``straggler's effect''. On the other hand, communication can be a primary bottleneck for federated settings because of clients bandwidth limitation and the possibility of server throttling.
    
    $\bullet$ Small models and redundancy. It is often
 the case that FL models are small because of clients limited computational power or memory, it
 therefore is unnecessary to train an FL model on billions of clients. Note that for optimization problem~\eqref{eq:objective-constrained} in many practical scenarios models have smaller size in comparison to the number of clients, this is because of the matroid constraint or simply because the size of the ground set is much smaller than the number of clients.

 Here we discuss our sampling scheme which crucially does not violate clients' privacy; see \Cref{alg:FCG}.

\noindent{\bf Unbiased client sampling scheme.} At each communication round the server chooses an active client from $i  \in [N]$ with probability $p_i$, and repeats this process $K$ times to obtain a 
multiset $A^{(t)}$ of size $K$  which may contain a client  more than once. Then the aggregation step is $\Delta^{(t)}=\frac{1}{K}\sum_{i\in A^{(t)}}\mb{v}_i^{(t)}$ where $\mb{v}_i^{(t)}$ defined in \eqref{eq:v}.

The next lemma shows that this sampling scheme is unbiased and in expectation the average update from chosen clients $A^{(t)}$ is equal to the average update from all clients.
\begin{lemma}[Unbiased sampling scheme]
\label{lem:unbiased-sampling}
    For Client Sampling Scheme, we have $\E_{A^{(t)}}\left[\langle \Delta^{(t)},\nabla\widehat{F}(\bx^{(t)}) \rangle \right] =   \sum_{i=1}^N p_i\langle\mb{v}_i^{(t)}, \nabla\widehat{F}(\mb{x}^{(t)}) \rangle$.
    % \[
    %     \E_{A^{(t)}}\left[\langle \Delta^{(t)},\nabla\widehat{F}(\bx^{(t)}) \rangle \right] =   \sum_{i=1}^N p_i\langle\mb{v}_i^{(t)}, \nabla\widehat{F}(\mb{x}^{(t)}) \rangle, ~~ \text{and}~~ \E_{A^{(t)}}\left[ \frac{1}{K}\sum_{i\in A^{(t)}} \nabla\widehat {f}_i(\mb{x}^{(t)})\right] =   \nabla\widehat{F}(\mb{x}^{(t)}) 
    % \]
\end{lemma}

The following lemma allows us to bound the variance which in turn helps to provide our convergence guarantees. To show the following result, it is required to upper bound the difference between the improvement on the function value by taking the direction suggested by a selected client versus taking the direction obtained by averaging all the directions from the clients. That is bounding $|\langle \mb{v}_{s},\nabla\widehat{F}(\bx^{(t)}) \rangle  - \langle  \sum_{i=1}^N \mb{v}_i^{(t)}, \nabla\widehat{F}(\mb{x}^{(t)}) \rangle|$, for a selected client $s$ in $A^{(t)}$. This in turn needs providing an upper bound for $\langle \mb{v}_i^{(t)}, \nabla \widehat{f_i}(\mb{x}^{(t)}) \rangle - \langle \mb{v}_s,\nabla\widehat{f}_i(\bx^{(t)}) \rangle$, for all $i\neq s$. At the heart, our proof relies on the properties of multilinear extensions of local submodular functions, the fact that each $\mb{v}_i^{(t)}$ corresponds to a base of the matroid, and Assumption~\ref{ass:bounded-gradients}.

\begin{lemma}[Bounded variance]
\label{lem:bounded-variance}
     Using Client Sampling Scheme we have $\Var(\langle \Delta^{(t)},\nabla\widehat{F}(\bx^{(t)}) \rangle)\leq 36r^2\gamma_t^2/K$.
        % \begin{align*}
        %     \Var(\langle \Delta^{(t)},\nabla\widehat{F}(\bx^{(t)}) \rangle)\leq 36r^2\gamma^2/K .
        % \end{align*}
\end{lemma}

Armed with the above lemmas and concentration inequalities e.g., Chebyshev's inequality, we can prove the convergence of Algorithm~\ref{alg:FCG}. This essentially is done by bounding the error introduced by the decentralized setting and carefully carrying the error through the analysis.

\begin{theorem}
\label{thm:convergence-Scheme2-Alg1}
    Let $\mc{M}$ be a matroid of rank $r$ and $\mc{P}$ be its matroid polytope. Using Client Sampling Scheme, for every $\eta,\delta > 0 $, Algorithm~\ref{alg:FCG} returns a $\mb{x}^{(T)}\in\mc{P}$ so that with probability at least $1-\delta$ 
    \begin{align*}
        &\left( 1- (1-\eta)^T\right)\texttt{OPT}\leq \\
        & ~~~~~ \widehat{F}(\mb{x}^{(T)}) +  \eta\left( r\sum_{t=1}^T\gamma_t+\frac{6r\sum_{t=1}^T\gamma_t}{\sqrt{K\delta/T }}\right)+\frac{T\eta^2r^2 m_F}{2}
    \end{align*}
    In particular, by setting $\eta=1/T$, Algorithm~\ref{alg:FCG} requires at most $\Tilde{O}(r)$ bits of communication per user per round ($\Tilde{O}(KTr)$ in total) and yields
    \begin{align*}
        \left( 1- 1/\mathrm{e}\right)\texttt{OPT} \leq \widehat{F}(\mb{x}^{(T)}) + \left( \frac{rD}{T}+\frac{6rD}{\sqrt{TK\delta }}+\frac{r^2 m_F}{2T}\right)
    \end{align*}
\end{theorem}
%Note that the additive errors of both theorems depend on the rank of the matroid, and by choosing $K > T$ other terms can be made negligible.

%% file: PracticalFCG.tex
\section{Practical Federated Continuous Greedy}\label{sec:fcg+}
One of the main considerations in federated optimization is the number of communication rounds. In this section, we show how Algorithm~\ref{alg:FCG} can be further improved to reduce communication rounds while simultaneously incorporating partial participation.

\noindent{\bf Algorithm description.}
Initially $\mb{x}^{(0)}=0$. On the $t$-th round of the \texttt{Practical Federated Continuous Greedy (\texttt{\textsc{Fed}CG+)}}, the central server first broadcasts the latest model $\mb{x}^{(t)}$ to a subset of active clients of size $K$ denoted by $A^{(t)}$. Next, each client $i \in A^{(t)}$ sets $\mb{x}_{i}^{(t,0)}=\mb{x}^{(t)}$ and performs $\tau$ steps of continuous greedy approach locally. More precisely, let $j\in\{0,1,\dots,\tau-1\}$ and $\mb{x}_i^{(t,j)}$ denote the $i$-th client's local model at communication round $t$ and local update step $j$, then the local updates are 
% as in equation \eqref{eq:update_local}. 
%\begin{wrapfigure}[5]{R}{0.6\textwidth}
\begin{align}
    \label{eq:update_local}
    &\widetilde{\mb{v}}_i^{(t,j)}\gets\argmax_{\mb{v}\in \mc{P}} \langle \mb{v},\nabla \widetilde{f}_i(\bx_i^{(t,j)},\zeta_i^{(t,j)})\rangle;\\
    &\mb{x}_i^{(t,j+1)} = \mb{x}_i^{(t,j)} +  \widetilde{\mb{v}}_i^{(t,j)}/\tau
\end{align}
% \begin{align}
% \label{eq:update_local}
% \begin{cases}
% \widetilde{\mb{v}}_i^{(t,j)}\gets\argmax_{\mb{v}\in \mc{P}} \langle \mb{v},\nabla \widetilde{f}_i(\bx_i^{(t,j)},\zeta_i^{(t,j)})\rangle\\
%     \mb{x}_i^{(t,j+1)} = \mb{x}_i^{(t,j)} +  \widetilde{\mb{v}}_i^{(t,j)}/\tau
% \end{cases}
% \end{align}
%\end{wrapfigure}
Here $\zeta_i^{(t,j)}$ is a set of subsets from the ground set $E$ sampled according to $\bx_i^{(t,j)}$, and $\nabla\widetilde{f}_i(\bx_i^{(t,j)},\zeta_i^{(t,j)})\in \mathbb{R}_{\geq 0}^n$ is an estimation of the gradient $\nabla\widehat{f}_i(\bx_i^{(t,j)})$ (more on this later).

After $\tau$ steps of local update, the $i$-th client from $A^{(t)}$ send her update $\widetilde{\Delta}_i^{(t+\tau)}=\bx_i^{(t,\tau)}-\bx_i^{(t,0)}$ to the secure aggregator to compute $\widetilde{\Delta}^{(t+\tau)}=\frac{1}{|A^{(t)}|}\sum_{i\in A^{(t)}} \widetilde{\Delta}_i^{(t+\tau)}$. Note that each $\widetilde{\Delta}_i^{(t+\tau)}$ belongs to $\mc{P}$ since it is a convex combination of vectors $\widetilde{\mb{v}}_i^{(t,j)}\in \mc{P}$. However, $\widetilde{\Delta}_i^{(t+\tau)}$ may not be an integral vector and in the worst case $\tilde{O}(n)$ bits are required to encode it.
After receiving $\widetilde{\Delta}^{(t+\tau)}$, the server produces the new global model with learning rate $\eta$:
\begin{align}
    \bx^{(t+\tau)}\gets \bx^{(t)} + \eta\widetilde{\Delta}^{(t+\tau)}
\end{align}
\noindent{\bf Gradient estimation.}
Evaluating the multilinear extension involves summing over all subsets $S$ of $E$, there are $2^{|E|}$ such subsets. However, recall $\partial\widehat{f}/\partial\mb{x}(e)=\E[f(R\cup\{e\})]-\E[f(R\setminus \{e\})]$
where $R\subseteq E$ is a random subset sampled according to $\mb{x}$. Hence a simple application of Chernoff's bound tells us by sampling sufficiently many subsets we can obtain a good estimation of $\nabla \widehat{f}(\mb{x})$ \cite{CalinescuCPV11,Vondrak08} (more details in the Appendix). For large enough $m$, let $\zeta_{i}^{(t,j)}=\{R_i^{(t,j,1)},\dots,R_i^{(t,j,m)}\}$ be subsets of $E$ that are sampled independently according to $i$-th client's local model $\mb{x}_i^{(t,j)}$. Let $\nabla\widetilde{f_i}(\mb{x}_i^{(t,j)},\zeta_{i}^{(t,j)})$ denote the \emph{stochastic} approximation of $\nabla\widehat{f_i}(\mb{x}_i^{(t,j)})$. Then with probability $1-\delta$ we have $\|\nabla\widehat{f_i}(\mb{x}_i^{(t,j)}) -\nabla\widetilde{f_i}(\mb{x}_i^{(t,j)},\zeta_{i}^{(t,j)})\|\leq \sigma$.

In convergence analyses of our algorithm there are two sources of randomness, one in the sampling schemes for client selection and the other in data sampling at each clients local data to estimate the gradients. In Algorithm~\ref{alg:FCG+} there are $\frac{T}{\tau}$ communications rounds between the server and clients. Define $\mc{I}_\tau=\{\tau i\mid i =1, 2, 3,\dots\}$ to denote the set of communication rounds with the server. 
Similar to Lemma~\ref{lem:unbiased-sampling}:

% \begin{lemma}[Bounded heterogeneity]
% \label{lem:bound-heterogeneity-Alg2}
%     Let $\mc{P}$ be a matroid polytope with rank $r$, $\mb{v}_1,\dots,\mb{v}_N\in \mc{P}$ and $\Bar{\mb{v}}=\frac{1}{N}\sum_{i=1}^N\mb{v}_i$. Under the bounded gradient assumption, for any $\bx \in [0, 1]^n$ we have
%     \[
%         \langle \Bar{\mb{v}} , \nabla\widehat{F}(\mb{x}) \rangle \geq \sum_{i=1}^N p_i\langle \mb{v}_i, \nabla \widehat{f_i}(\mb{x}) \rangle - 2\gamma\sqrt{r}
%     \]
% \end{lemma}

\begin{lemma}[Unbiased sampling scheme]
\label{lem:unbiased-sampling-Alg2}
    For Client Sampling Scheme, at every communication round $t+\tau\in\mc{I}_\tau$, we have $\E_{A^{(t)}}\left[\langle \widetilde{\Delta}^{(t+\tau)},\nabla\widehat{F}(\bx^{(t)}) \rangle \right] =  \langle \sum_{i=1}^N p_i \widetilde{\Delta}_i^{(t+\tau)}, \nabla\widehat{F}(\mb{x}^{(t)}) \rangle$.
    % \begin{align*}
    %     \E_{A^{(t)}}\left[\langle \widetilde{\Delta}^{(t+\tau)},\nabla\widehat{F}(\bx^{(t)}) \rangle \right] =  \langle \sum_{i=1}^N p_i \widetilde{\Delta}_i^{(t+\tau)}, \nabla\widehat{F}(\mb{x}^{(t)}) \rangle
    % \end{align*}
\end{lemma}
% \begin{proof}
%     The proof is almost identical to the one in Lemma~\ref{lem:unbiased-sampling}. \anote{make sure to add the proofs}
% \end{proof}
Bounding the variance of $\langle \widetilde{\Delta}^{(t+\tau)},\nabla\widehat{F}(\bx^{(t)}) \rangle$ is much more delicate than in Lemma~\ref{lem:bounded-variance}. The are two main reasons, one is the deviation caused by local steps, second is the deviation from the true  $\nabla\widehat{F}(\mb{x}^{(t)})$ caused by gradient estimations. To handle the deviation cause by local steps we assume all $\widehat{f}_i$ have $L$-Lipschitz continuous gradients i.e., $\|\nabla\widehat{f}_i(\mb{x})-\nabla\widehat{f}_i(\mb{y})\|\leq L\|\mb{x}-\mb{y}\|$, $\forall\mb{x},\mb{y}\in\mc{P}$. In fact as shown in Lemma 3 of \cite{MokhtariHK18}, for each $\widehat{f}_i$ and local model $\mb{x}_i^{(t,j+1)}=\bx_i^{(t,j)} + \widetilde{\mb{v}}_i^{(t,j)}/\tau$, it holds that
\begin{align*}
    \left\| \nabla\widehat{f}_i(\mb{x}_i^{(t,j+1)})-\nabla\widehat{f}_i(\bx_i^{(t,j)}) \right\| 
    & \leq  \frac{m_{f_i}\sqrt{r}}{\tau}\left\|\widetilde{\mb{v}}_i^{(t,j)}\right\| \\
    & \leq  \frac{m_{f_i} r}{\tau}
\end{align*}
The factor $-m_{f_i}$ is in fact a lower bound on the entries of the Hessian matrix. That is $\frac{\partial \widehat{f}_i}{\partial \mb{x}(i)\partial \mb{x}(j)}\geq -m_{f_i}$ \cite{HassaniSK17-SGD-submodular}. 

Let $Q=\sum_{t\in I_\tau} L_t$ where $L_t$ is such that for all $i\in[N]$ and $j\in[\tau]$ it holds
\begin{align*}
    \left \| \nabla \widehat{f_i}(\bx_i^{(t,j)})  - \nabla \widehat{f_i}(\bx^{(t)}) \right \| \leq L_t\| \mb{x}_i^{(t,j)} - \mb{x}^{(t)} \| 
\end{align*}
Observe that $L_t$ is upper bounded by $\max_{i} m_{f_i}\sqrt{r}$.
\begin{lemma}[Bounded variance]
\label{lem:bounded-variance-Alg2}
    For $t+\tau\in\mc{I}_\tau$ we have $\Var(\langle \widetilde{\Delta}^{(t+\tau)},\nabla\widehat{F}(\bx^{(t)}) \rangle)\leq\frac{1}{K}\left( 6r\gamma_t + 2(\sigma r+L_tr^{1.5}) \right)^2$.
        % \begin{align*}
        %     \Var(\langle \widetilde{\Delta}^{(t+\tau)},\nabla\widehat{F}(\bx^{(t)}) \rangle)\leq\left( 6r\gamma + (2\sigma+4L)\sqrt{\nicefrac{r^3}{n}}\right)^2/K
        % \end{align*}    
\end{lemma}

While the variance can be made arbitrary small by sampling more clients, the additive error caused by local steps cannot be controlled by sampling more clients. This shows up in the next theorem.

\begin{algorithm}[tb]
  \caption{Practical FedCG \texttt{(\textsc{Fed}CG+)}}
  \label{alg:FCG+}
  \begin{algorithmic}[1]
    \STATE {\bfseries Input: }{Matroid polytope $\mc{P}$, number of communication rounds $T/\tau$, server's learning rate $\eta$, $\sigma,\delta>0$, and $K$.}
    \STATE $\bx^{(0)}=\mb{0}$, $m=O(\log{(TK/\delta)}/\sigma^2)$
    \FOR{$t=0,\tau, 2\tau,\dots, (T-1)/\tau$}
        \STATE Server selects a subset of $K$ \emph{active} clients $A^{(t)}$ according to Client Sampling Scheme, and sends $\bx^{(t)}$ to them.
        \FOR{Client $i$ in $A^{(t)}$ in parallel }
            \STATE $\bx_i^{(t,0)}\gets \bx^{(t)}$
            \FOR{$j=0,\dots,\tau-1$}
                \STATE Randomly sample $m$ sets $\zeta_{i}^{(t,j)}=\{R_i^{(t,j,1)},\dots,R_i^{(t,j,m)}\}$ according to $\mb{x}_i^{(t,j)}$ 
                \STATE Let $\nabla\widetilde{f}_i(\bx_i^{(t,j)},\zeta_{i}^{(t,j)})$ be the estimate of $\nabla \widehat{f}_i(\bx_i^{(t,j)})$ 
                \STATE $\widetilde{\mb{v}}_i^{(t,j)}\gets\argmax_{\mb{v}\in \mc{P}} \langle \mb{v},\nabla\widetilde{f}_i(\bx_i^{(t,j)},\zeta_{i}^{(t,j)})\rangle$
                \STATE $\bx_i^{(t,j+1)}\gets \bx_i^{(t,j)} + \widetilde{\mb{v}}_i^{(t,j)}/\tau$
            \ENDFOR
            \STATE $\widetilde{\Delta}_i^{(t+\tau)}\gets \bx_i^{(t,\tau)}-\bx_i^{(t,0)}$ \COMMENT{Local model change}
            \STATE Send $ \widetilde{\Delta}_i^{(t+\tau)}$ back to the secure aggregator.
        \ENDFOR
        \STATE SecAgg: $\widetilde{\Delta}^{(t+\tau)}=\frac{1}{|A^{(t)}|}\sum_{i\in A^{(t)}} \widetilde{\Delta}_i^{(t+\tau)}$.
        \STATE Server updates: $\bx^{(t+\tau)}\gets \bx^{(t)} + \eta\widetilde{\Delta}^{(t+\tau)}$
    \ENDFOR
    %\STATE {\bfseries Output: } $\bx^{(T)}$
    \STATE Apply a proper rounding scheme on $\bx^{(T)}$ to obtain a solution for \eqref{eq:objective-constrained}
  \end{algorithmic}
\end{algorithm}

\begin{theorem}
\label{thm:convergence-Alg2-Scheme2}
    Let $\mc{M}$ be a matroid of rank $r$ and $\mc{P}$ be its matroid polytope. Using Client Sampling Scheme, for every $\eta,\delta > 0$, Algorithm~\ref{alg:FCG+} returns a $\mb{x}^{(T/\tau)}\in\mc{P}$ such that with probability at least $1-\delta$ it holds 
    \begin{align*}
        &( 1- (1-\eta)^{T/\tau})\texttt{OPT} \\
        & \quad\quad \leq \widehat{F}(\mb{x}^{(T/\tau)})+\frac{T\eta^2r^2 m_F}{2\tau}\\ 
        & \quad\quad +(\overbrace{\eta r\sum_{t=1}^{T/\tau}\gamma_t}^{\text{heterogeneity}} +
         \overbrace{2\sigma r+2\eta r^{1.5}\sum_{t=1}^{T/\tau} L_t )}^{\text{local steps}}\\
         &\quad \quad+\overbrace{{\sqrt{T}(6\eta r\sum_{t=1}^{T/\tau} \gamma_t + 2\sigma r+2\eta r^{1.5}\sum_{t=1}^{T/\tau}L_t) )}/{\sqrt{K\tau\delta }}}^{\text{client sampling}})
    \end{align*}
    % \begin{multline*}
    %      ( 1- (1-\eta)^{T/\tau})\texttt{OPT} \leq \widehat{F}(\mb{x}^{(T/\tau)})+\frac{T\eta^2r^2 m_F}{2\tau}+\\(\overbrace{\eta r\sum_{t=1}^{T/\tau}\gamma_t}^{\text{heterogeneity}} +
    %      \overbrace{2\sigma r+2\eta r^{1.5}\sum_{t=1}^{T/\tau} L_t )}^{\text{local steps}}+\\
    %      \overbrace{{\sqrt{T}(6\eta r\sum_{t=1}^{T/\tau} \gamma_t + 2\sigma r+2\eta r^{1.5}\sum_{t=1}^{T/\tau}L_t) )}/{\sqrt{K\tau\delta }}}^{\text{client sampling}})
    % \end{multline*}
       
 In particular, for $\eta = \tau/T$, Algorithm~\ref{alg:FCG+} has at most $\Tilde{O}(n)$ bits of communication per user per round ($\Tilde{O}(KTn/\tau)$ in total) and yields
    \begin{align*}
        \left( 1- 1/\mathrm{e}\right)\texttt{OPT} 
        &\leq \widehat{F}(\mb{x}^{(T)}) +\frac{\tau r^2 m_F}{2T} \\ 
        &+ \frac{rD\tau}{T}+2\sigma r(1+\frac{\sqrt{\tau}}{\sqrt{KT \delta}})+\frac{2r^{1.5}Q\tau}{T} \\
        &+\frac{\sqrt{\tau}(6Dr+2r^{1.5}Q)}{\sqrt{TK\delta}}
    \end{align*}
\end{theorem}

\begin{remark}
    From a fractional solution $\mb{x}^{(T)}\in \mc{P}$ returned by \Cref{alg:FCG,alg:FCG+} one can obtain a solution for \eqref{eq:objective-constrained} using appropriate rounding schemes. Rounding schemes such as pipage rounding \cite{CalinescuCPV11}, swap rounding \cite{ChekuriVZ10-swap-rounding}, or greedy rounding are oblivious; they do not require access to the objective function. Therefore, the server can utilize these rounding schemes without needing to know the decomposable function itself.
\end{remark}

%% file: Discrete_examples.tex
\section{Discrete Algorithm in Federated Setting}
\label{sec:discrete-fl}

\begin{algorithm}[tb]
  \caption{Federated Discrete Greedy}
  \label{alg:Discrete-FL}
  \begin{algorithmic}[1]
    \STATE {\bfseries Input: }{Matroid $\mc{M}$ of rank $r$, importance factors $\{w_i\}_{i=1}^N$, $\varepsilon \in (0,1)$.}
    %\STATE $S\gets\emptyset$, $\kappa\gets \frac{3}{\varepsilon^2}\log(2r n^{r}/\delta)$
    \STATE $S\gets\emptyset$, $\kappa\gets \tilde{O}(rn/\varepsilon^2)$
    \FOR{$t=0$ to $r-1$}
        \STATE Server sends $S$ to all clients.
        \FOR{each client $i$ in parallel}
            \STATE $\kappa_i\gets\min\{\kappa\cdot w_i, 1\}$
            \FOR{all $e$ such that $S\cup\{e\}\in \mc{M}$}
               \STATE $\Delta_i[e]\gets (f_i(S\cup \{e\})-f_i(S))/\kappa_i$
            \ENDFOR 
            
            % \STATE Computes vector $\Delta_i$:
            % \begin{align*}
            %     \Delta_i[e]=(f_i(S\cup \{e\})-f_i(S))/\kappa_i; \quad \forall e\in E\setminus S
            % \end{align*}            
            \COMMENT{/* Randomized Response */}
            \STATE With probability $\kappa_i$ sends $\Delta_i$ to the secure aggregator
            \STATE With probability $1- \kappa_i$ does nothing
        \ENDFOR
        \STATE SecAgg: $\Delta^{(t)}$ (the sum of $\Delta_i$ received in this round.)
        \STATE Server updates: $S\gets S\cup \{\argmax\limits_{e:S\cup\{e\}\in \mc{M}}\Delta^{(t)}[e]\}$
    \ENDFOR
    \STATE {\bfseries Output: } $S$
  \end{algorithmic}
\end{algorithm}

\noindent While parallel SGD and continuous methods such as ours in this paper are commonly used as the main tool in federated optimization, we introduce a rather discrete approach to the field and believe it will find further applications. Our approach is inspired by recent works of \citet{Sparsification-RafieyY22,Kenneth-CutSparsifier} on a seemingly unrelated topic. \citet{Sparsification-RafieyY22} introduced a method to sparsify a sum of submodular functions in a centralized setting which was improved by \cite{Kenneth-CutSparsifier}. We tailor their approach to the federated setting and discuss its effectiveness for discrete problems such as \texttt{Facility Location} and \texttt{Maximum Coverage} problems. 
%Recall the optimization problem \eqref{eq:objective-constrained} 
% \[
% \label{eq:objective-discrete}
%     \max_{S\in \mc{I}}\left\{F(S)=\sum_{i=1}^N f_i(S)\right\}
% \]
%where $\mc{M}=(E,\mc{I})$ is a matroid of rank $r$. 
At the heart of this approach is for clients to know their ``importance'' without sharing sensitive information. In the monotone case, the  \emph{importance factor} for client $i$ is defined as:
\[
    w_i = \max_{e\in E} \frac{f_i(\{e\})}{F(\{e\})}.
\]
In several cases such as \texttt{Max Facility Location} and \texttt{Maximum Coverage} computing the importance factor can be done efficiently and with constant number of communication rounds. For now, let us continue by assuming each client knows its own importance factor.

%In general, computing importance factors is computationally expensive even in a centralized manner. However, this can be done efficiently and with constant number of communication rounds for various problems such as \texttt{Max Facility Location} and \texttt{Maximum Coverage}. For now, let us continue by assuming each client knows its own importance factor. 

% The first step of our algorithm is a pre-processing step to compute importance factors for all clients in a federated manner. This may not always be possible due to the definition of importance factor. However, we will discuss that this step can be efficiently executed in federated setting for several cases.  

\noindent{\bf Algorithm~\ref{alg:Discrete-FL} description.}
Let $\varepsilon\in(0,1)$ and set $S=\emptyset$. The server gradually adds elements to $S$ for only $r$ rounds. At round $t$ the central server broadcasts the current set $S$ to all clients (or a subset of active clients). Then each client $i$ computes the marginal contribution of each element from $E\setminus S$ to its local function $f_i$
\begin{align*}
    \Delta_i[e]=f_i(S\cup \{e\})-f_i(S); \quad \forall e: S\cup\{e\}\in \mc{I}
\end{align*}
Let $\kappa=\tilde{O}(rn/\varepsilon^2)$. Then each client sends its update in a \emph{randomized response} manner, that is, the $i$-th client with probability $\kappa_i=\min\{1,\kappa\cdot w_i\}$ sends the scaled vector $(\frac{1}{\kappa_i})\Delta_i$ to the secure aggregator, and with the complement probability does not send anything. The secure aggregator computes $\Delta^{(t)}$; the sum of the update vectors it has received. The server updates $S\gets S\cup \{\argmax_{e:S\cup\{e\}\in \mc{M}}\Delta^{(t)}[e]\}$.
% \begin{align}
%     S\gets S\cup \{\argmax_{e:S\cup\{e\}\in \mc{M}}\Delta^{(t)}[e]\}
% \end{align}
The following theorem follows from the sparsification results in \cite{Sparsification-RafieyY22,Kenneth-CutSparsifier} and approximation guarantees of the greedy algorithm \cite{NemhauserWF78}. 
\begin{theorem}
\label{thm:convergence-discrete}
    For every $\varepsilon \in (0,1)$, Algorithm~\ref{alg:Discrete-FL}, with probability at least $1-1/n$, returns a subset $S\in \mc{I}$ such that
         $({1}/{2}-\varepsilon)\texttt{OPT} \leq F(S).$ %Moreover, at each round $\tilde{O}(rn^2/\varepsilon^2)$ clients participate.

         Under a cardinality constraint (uniform matroid), Algorithm~\ref{alg:Discrete-FL}, with probability at least $1-1/n$, returns a subset $S\in \mc{I}$ such that
         $(1-{1}/{\mathrm{e}}-\varepsilon)\texttt{OPT} \leq F(S).$ 

         Moreover, in expectation, at each round $\tilde{O}(rn^2/\varepsilon^2)$ clients participate where $r$ is the rank of the matroid.
\end{theorem}

%\subsection{Discrete examples} \label{sec:discrete}

% We consider two well-studied problems, namely \texttt{Facility Location} and \texttt{Maximum Coverage} problems and discuss the details of how to compute importance factors in federated setting efficiently and prove that the expected number of clients participating in each round is small for these two problems.

\noindent{\bf Facility Location Problem.}
Let $\mc{C}$ be a set of $N$ clients and $E$ be a set of facilities with $|E|=n$. For $c:\mc{C}\times E\to \zR$ let the $i$-th client's score function over a subset of facilities be $f_i(A)=\max_{j\in A} c(i,j)$. The objective for \texttt{Max Facility Location} is $\max_{S\subseteq E,|S|\leq k} \sum_{i=1}^N \max_{j\in S} c(i,j)$.
% \begin{align*}
%     \max_{S\subseteq E,|S|\leq k} \left\{ F(A)=\sum_{i=1}^N \max_{j\in A} c(i,j)\right\}
% \end{align*}
For each client $i$  the importance factor is $w_i = \max_{j\in E}\frac{c(i,j)}{ F(\{j\})}$. 

% \begin{align}
% w_i = \max\limits_{A\subseteq E}\frac{f_i(A)}{F(A)}=\max\limits_{\substack{A\subseteq E}}\frac{\max\limits_{j\in A}c(i,j)}{F(A)} =  \max\limits_{j\in E}\frac{c(i,j)}{ F(\{j\})}.
% \end{align}

In several situations computing the importance factors is straightforward. For instance, in movie recommendation systems where a facility location objective is used (see \Cref{sec:discrete}), the average rating i.e., $F(\{j\})$, for all movies are publicly available. Hence, each client knows its own importance factor. Having $w_i$, we can apply Algorithm~\ref{alg:Discrete-FL}. Doing so, note that at each round the expected number of clients participating is $\tilde{O}(kn^2/\varepsilon^2)$, \textbf{independent of $N$}. 

In other situations, clients can compute their importance factors in a federated setting, using a secure aggregator and without sharing their data with other clients. Further details on this and discussions about the \texttt{Max Coverage} problem are presented in \Cref{sec:discrete}.

%% file: Conclusion.tex
\section{Conclusions and Future Work}
We present \texttt{\textsc{Fed}CG}, the first algorithm for decomposable submodular maximization in federated setting under matroid constraints. \texttt{\textsc{Fed}CG} is based on the continuous greedy algorithm and achieves the best possible approximation factor i.e., $1-1/\mathrm{e}$ under mild assumptions even faced with client selection and low communication rounds. Additionally, we introduce a new federated framework for discrete problems. 
%which has very few communication rounds and is computationally inexpensive for clients. 
Our work leads to many interesting directions for future work, such as providing stronger privacy guarantee using differentially private methods \cite{MitrovicB0K17,RafieyY20,ChaturvediNZ21}, finding a lower bound on the bit complexity, and improving the additive errors of our algorithms.

%% file: Appendix-Full-participation.tex
\section{Preliminary Results and Proof of Convergence for Full Participation}

\paragraph{Quick overview.} At the heart of the analyses for the approximation ratio of the \texttt{Centralized Continuous Greedy} is to show that at each iteration the algorithm reduces the gap to the optimal solution by a significant amount \cite{CalinescuCPV11}. We follow the same general idea although there are several subtleties that we should address. Mainly, it is required to compare the improvement we obtain by taking direction $\mb{v}$ versus the improvement obtained by taking direction $\sum_{i=1}^N p_i\mb{v}_i$ where $\mb{v}=\argmax_{\mb{y}\in\mc{P}}\langle \mb{y},\nabla\widehat{F}(\mb{x})\rangle$ and  $\mb{v}_i=\argmax_{\mb{y}\in\mc{P}}\langle \mb{y},\nabla\widehat{f}_i(\mb{x})\rangle$. To establish a comparison between the two, we need several intermediate lemmas.

We start off by focusing on providing an upper bound on $\texttt{OPT}-\widehat{F}(\mb{x})$. 
\begin{lemma}
\label{lem:full-participation}
    Let $F=\sum_{i=1}^N p_i f_i$ be a function where each $f_i:2^E\to\mathbb{R}_{+}$ is a monotone submodular function. Suppose $\mc{P} \subseteq [0, 1]^n$ is a polytope and define $\texttt{OPT} = \max_{\bx\in \mc{P}} \widehat{F}(\bx)$. Then for any $\bx \in [0, 1]^n$ and $\mb{v}_i=\argmax_{\mb{y}\in \mc{P}}\langle \mb{y},\nabla\widehat{f}_i(\mb{x})\rangle$ we have 
    \[
        \texttt{OPT}- \widehat{F}(\bx)  \leq  \sum_{i=1}^N p_i \langle\mb{v}_i, \nabla \widehat{f_i}(\mb{x}) \rangle.
    \]
\end{lemma}

\begin{proof}[Proof of Lemma~\ref{lem:full-participation}]
    First, let us derive some observations on the decomposable submodular function $F$ and its multilinear extension $\widehat{F}$. By definition $\widehat{F}(\mb{x})=\E_{X\sim\mb{x}}[F(X)]$. Then
    \begin{align*}
        \widehat{F}(\mb{x}) =\E_{X\sim\mb{x}}[F(X)]
        = \E_{X\sim\mb{x}}\left[\sum_{i=1}^N p_i f_i\right] 
        = \sum_{i=1}^N p_i(\E_{X\sim\mb{x}}[f_i(X)]) 
        = \sum_{i=1}^N p_i\widehat{f_i}(\mb{x})
    \end{align*}
    Then the gradient of $\widehat{F}$ at any point is equal to the average of gradients from local functions. That is 
    $
        \nabla\widehat{F}(\mb{x})=\nabla\left( \sum_{i=1}^N p_i\widehat{f_i}(\mb{x})\right)=\sum_{i=1}^N p_i\nabla \widehat{f_i}(\mb{x}).
    $
    % Finally, linearity of dot product yields,
    % \begin{align*}
    %     \langle \mb{v}, \nabla \widehat{F}(\mb{x}) \rangle = \langle \mb{v}, \sum_{i=1}^N p_i\nabla \widehat{f_i}(\mb{x}) \rangle =\sum_{i=1}^N p_i \langle \mb{v}, \nabla \widehat{f_i}(\mb{x}) \rangle
    % \end{align*}

    Now let $\mb{w}\in \mc{P}$ be such that $\widehat{F}(\mb{w})=\texttt{OPT}$. 
    \begin{align*}
        % \max_{\mb{v}\in\mc{P}}\langle \mb{v}, \nabla \widehat{F}(\mb{x}) \rangle = \langle \sum\limits_{i=1}^N p_i \mb{v}_i, \nabla \widehat{F}(\mb{x}) \rangle
        &\langle \mb{w}, \nabla \widehat{F}(\mb{x}) \rangle \leq   \max_{\mb{v}\in \mc{P}}\langle \mb{v}, \nabla \widehat{F}(\mb{x}) \rangle = \max_{\mb{v}\in \mc{P}}\langle \mb{v},  \sum_{i=1}^N p_i\nabla\widehat{f}_i(\mb{x}) \rangle = \max_{\mb{v}\in \mc{P}} \sum_{i=1}^N  p_i\langle \mb{v}, \nabla \widehat{f_i}(\mb{x}) \rangle \\
        &\leq \sum_{i=1}^N p_i\max_{\mb{y}\in \mc{P}} \langle \mb{y}, \nabla \widehat{f_i}(\mb{x}) \rangle = \sum_{i=1}^N p_i \langle \mb{v}_i, \nabla \widehat{f_i}(\mb{x}) \rangle
    \end{align*}
    In what follows, we prove $\texttt{OPT}- \widehat{F}(\bx) \leq \langle\mb{w},\nabla \widehat{F}(\bx)\rangle$. Define $\mb{d}=(\mb{x}\vee \mb{w})-\mb{x}=(\mb{w}-\mb{x})\vee \mb{0}$. By the monotonicity, we have
    \begin{align*}
        \texttt{OPT}=\widehat{F}(\mb{w}) \leq \widehat{F}(\mb{x}\vee\mb{w}).
    \end{align*}
    Note that $\mb{d}>\mb{0}$, and hence by concavity of $\widehat{F}$ along any positive direction we get
    \begin{align*}
        \widehat{F}(\mb{x}\vee\mb{w})=\widehat{F}(\mb{x}+\mb{d}) \leq \widehat{F}(\mb{x})+\langle \mb{d},\nabla\widehat{F}(\mb{x})\rangle.
    \end{align*}
    Combining the above inequalities we obtain
    \begin{align*}
        \texttt{OPT}-\widehat{F}(\mb{x})\leq\widehat{F}(\mb{x}\vee\mb{w})-\widehat{F}(\mb{x})\leq\langle \mb{d},\nabla\widehat{F}(\mb{x})\rangle 
    \end{align*}
    Now, since $\nabla\widehat{F}(\mb{x})$ is nonnegative and $\mb{d}\leq \mb{w}$ then we get
    \begin{align*}
        \texttt{OPT}-\widehat{F}(\mb{x})\leq \langle \mb{w},\nabla\widehat{F}(\mb{x})\rangle
    \end{align*}
    This completes the proof.
\end{proof}

The following is an immediate consequence of Lemma~\ref{lem:full-participation}.
\begin{corollary}
\label{cor:opt-F}
    For any polytope $\mc{P} \subseteq [0, 1]^n$, $\bx^{(t)} \in [0, 1]^n$, and $\mb{v}_i^{(t)}=\argmax_{\mb{y}\in \mc{P}} \langle \mb{y}, \nabla \widehat{f_i}(\mb{x}^{(t)}) \rangle$,
    $
        \texttt{OPT} - \widehat{F}(\mb{x}^{(t)}) \leq  \sum_{i=1}^N p_i \langle\mb{v}_i^{(t)}, \nabla\widehat{f_i}(\mb{x}^{(t)}) \rangle.
    $
\end{corollary}

Next we use Assumption~3.1 to bound the difference between $\sum_{i=1}^N p_i\langle \mb{v}_i^{(t)}, \nabla \widehat{f_i}(\mb{x}^{(t)}) \rangle$ and $\sum_{i=1}^N p_i\langle \mb{v}_i^{(t)}, \nabla \widehat{F}(\mb{x}^{(t)}) \rangle$ at every iteration $t$.

\begin{lemma}[Bounded heterogeneity]
\label{lem:bound-heterogeneity}
    Let $\mc{M}$ be a matroid of rank $r$ and $\mc{P}$ be its corresponding matroid polytope. For any $\bx^{(t)} \in [0, 1]^n$ let $\mb{v}_i^{(t)}=\argmax_{\mb{y}\in \mc{P}}\langle \mb{y},\nabla\widehat{f}_i(\mb{x}^{(t)})\rangle$ and $\Bar{\mb{v}}^{(t)}=\sum_{i=1}^N p_i\mb{v}_i^{(t)}$. Then, under the bounded gradient dissimilarity assumption 3.1, we have
    \[
        \langle \Bar{\mb{v}}^{(t)} , \nabla\widehat{F}(\mb{x}^{(t)}) \rangle \geq \sum_{i=1}^N p_i\langle \mb{v}_i^{(t)}, \nabla \widehat{f_i}(\mb{x}^{(t)}) \rangle - r\gamma_t
    \]
\end{lemma}

\begin{proof}[Proof of Lemma~\ref{lem:bound-heterogeneity}]
    The proof is straightforward.
    \begin{align*}
        &\sum_{i=1}^N p_i\langle \mb{v}_i^{(t)}, \nabla \widehat{f_i}(\mb{x}^{(t)}) \rangle - \langle \Bar{\mb{v}}^{(t)} , \nabla\widehat{F}(\mb{x}^{(t)}) \rangle
         = \sum_{i=1}^N p_i\langle \mb{v}_i^{(t)}, \nabla \widehat{f_i}(\mb{x}^{(t)}) \rangle - \sum_{i=1}^N p_i\langle \mb{v}_i^{(t)}, \nabla\widehat{F}(\mb{x}^{(t)}) \rangle \\
        &= \sum_{i=1}^N p_i\langle \mb{v}_i^{(t)}, \nabla \widehat{f_i}(\mb{x}^{(t)}) - \nabla\widehat{F}(\mb{x}^{(t)}) \rangle 
        \leq \sum_{i=1}^N p_i r\gamma_t = r\gamma_t
    \end{align*}
where in the last inequality we used the fact that each $\mb{v}_i^{(t)}$ corresponds to a base in the matroid and we have $\mb{v}_i^{(t)}\in \{0,1\}^n$ and $|\mb{v}_i^{(t)}|_1=r$. Given this and Assumption~3.1 we have $|\langle \mb{v}_i^{(t)}, \nabla \widehat{f_i}(\mb{x}^{(t)}) - \nabla\widehat{F}(\mb{x}^{(t)}) \rangle|\leq r\gamma_t$. 
\end{proof}

\subsection{Putting Everything Together: Proof of \Cref{thm:full-participation}}
We now are ready to prove \Cref{thm:full-participation}. 

\begin{proof}[Proof of \Cref{thm:full-participation}]
    Recall that $\bx^{(t+1)}\gets \bx^{(t)} + \eta\Delta^{(t)}$ where in the full participation case $\Delta^{(t)}=\sum_{i=1}^N p_i \mb{v}_i^{(t)}$. According to the Taylor's expansion of the function $\widehat{F}$ near the point $\bx^{(t)}$ we can write

    \begin{align}
        \widehat{F}(\mb{x}^{(t+1)})
         &= \widehat{F}(\bx^{(t)}) + \langle \mb{x}^{(t+1)} - \bx^{(t)},\nabla\widehat{F}(\bx^{(t)}) \rangle + \frac{1}{2} \langle \mb{x}^{(t+1)} - \bx^{(t)}, \mb{H}_{\widehat{F}} (\mb{x}^{(t+1)} - \bx^{(t)})\rangle\\
         \label{eq:Taylor}
         &= \widehat{F}(\bx^{(t)}) + \eta\langle \Delta^{(t)},\nabla\widehat{F}(\bx^{(t)}) \rangle + \frac{\eta^2}{2} \langle\Delta^{(t)}, \mb{H}_{\widehat{F}} \Delta^{(t)} \rangle
    \end{align}
    where $\mb{H}_{\widehat{F}}$ is the Hessian matrix i.e., the second derivative matrix. We provide a lower bound on each entry of $\mb{H}_{\widehat{F}}$. By the result of \cite{CalinescuCPV11} and definition of the multilinear extension:
    \begin{align}
        \frac{\partial \widehat{F}}{\partial \mb{x}(i)\partial \mb{x}(j)} 
        &= \E_{R\sim \mb{x}}[F(R\cup\{i,j\})-F(R\cup\{i\}\setminus\{j\})] -\E_{R\sim \mb{x}}[F(R\cup\{j\}\setminus\{i\}) - F(R\setminus\{i,j\})]\\
        &\geq -\max\{F(\{i\}),F(\{j\})\} \geq -\max_{e\in E}F(\{e\}) = -m_F 
    \end{align}
    where the second last inequality is  a direct consequence of the submodularity of $F$. This means every entry of the Hessian is at least $-m_F$. Thus, we arrive at the following lower bound
    \begin{align}
    \label{eq:F-lip}
        \langle\Delta^{(t)}, \mb{H}_{\widehat{F}} \Delta^{(t)} \rangle \geq \sum_{i=1}^n\sum_{j=1}^n \Delta^{(t)}(i)\Delta^{(t)}(j) \mb{H}_{\widehat{F}}{(i,j)} \geq -m_F \sum_{i=1}^n\sum_{j=1}^n \Delta^{(t)}(i)\Delta^{(t)}(j)= -m_F\left(\sum_{i=1}^n \Delta^{(t)}(i)\right)^2 \geq -m_F r^2   
    \end{align}
    where in the last inequality we used the fact that $\Delta^{(t)}\in \mc{P}$ as it is a convex combination of vectors from $\mc{P}$, and hence $\sum_{i=1}^n \Delta^{(t)}(i)\leq r$. 

    Thus from \eqref{eq:Taylor} and \eqref{eq:F-lip} it follows that  
    
    \begin{align}
    \label{eq:delta}
        \widehat{F}(\mb{x}^{(t+1)})
         &\geq \widehat{F}(\bx^{(t)}) + \eta\langle \Delta^{(t)},\nabla\widehat{F}(\bx^{(t)}) \rangle - \frac{\eta^2\cdot m_F r^2}{2}
        %  \\
        %  &= \widehat{F}(\bx^{(t)}) + \eta\left( \langle \frac{1}{N}\sum_{i=1}^N \mb{v}_i^{(t)}, \nabla\widehat{F}(\mb{x}^{(t)}) \rangle\right)  - \frac{C\cdot \eta^2}{2}\\
        %  &\geq \widehat{F}(\bx^{(t)}) + \eta\left(  \sum_{i=1}^N \frac{1}{N}\langle\mb{v}_i^{(t)}, \nabla\widehat{f_i}(\mb{x}^{(t)}) \rangle - 2\gamma\sqrt{r}\right)  - \frac{C\cdot \eta^2}{2}\tag{by Lemma~\ref{lem:bound-heterogeneity}} \\
        % & \geq \widehat{F}(\bx^{(t)}) + \eta\left( \texttt{OPT} - \widehat{F}(\mb{x}^{(t)}) \right) - 2\eta\gamma\sqrt{r}- \frac{C\cdot \eta^2}{2}\tag{by Corollary~\ref{cor:opt-F}}
    \end{align}

    It is now required to provide a lower bound for $\langle \Delta^{(t)}, \nabla\widehat{F}(\bx^{(t)} \rangle$  in terms of $\texttt{OPT}$. We prove the following lemma.
    \begin{lemma}
    \label{lem:bound-heterogeneity-2}
        Let $\gamma_t$ be as in  Assumption~\ref{ass:bounded-gradients}. Then we have
        $\texttt{OPT}-\widehat{F}(\mb{x}^{(t)}) \leq \langle \Delta^{(t)}, \nabla\widehat{F}(\bx^{(t)} \rangle + r\gamma_t$.
    \end{lemma}
    \begin{proof}[Proof of Lemma~\ref{lem:bound-heterogeneity-2}]
        By Corollary~\ref{cor:opt-F} we have that 
        \[
        \texttt{OPT} - \widehat{F}(\mb{x}^{(t)}) \leq  \sum_{i=1}^N p_i \langle\mb{v}_i^{(t)}, \nabla\widehat{f_i}(\mb{x}^{(t)}) \rangle
        \]
        
        Furthermore, by Lemma~\ref{lem:bound-heterogeneity}, we have 
        \[
        \sum_{i=1}^N p_i\langle \mb{v}_i^{(t)}, \nabla \widehat{f_i}(\mb{x}^{(t)}) \rangle - \langle \Delta^{(t)} , \nabla\widehat{F}(\mb{x}^{(t)}) \rangle\leq r\gamma_t
        \] 
        These two give the desired result.
    %     \begin{align*}
    %     &\sum_{i=1}^N p_i\langle \mb{v}_i^{(t)}, \nabla \widehat{f_i}(\mb{x}) \rangle - \langle \Delta^{(t)} , \nabla\widehat{F}(\mb{x}) \rangle
    %      = \sum_{i=1}^N p_i\langle \mb{v}_i^{(t)}, \nabla \widehat{f_i}(\mb{x}) \rangle - \sum_{i=1}^N p_i\langle \mb{v}_i^{(t)}, \nabla\widehat{F}(\mb{x}) \rangle \\
    %     &= \sum_{i=1}^N p_i\langle \mb{v}_i^{(t)}, \nabla \widehat{f_i}(\mb{x}) - \nabla\widehat{F}(\mb{x}) \rangle 
    %     \leq \sum_{i=1}^N p_i r\gamma = r\gamma
    % \end{align*}
    % where in the last inequality we used the fact that for each $\mb{v}_i^{(t)}$ we have $\mb{v}_i^{(t)}\in \{0,1\}^n$ and $|\mb{v}_i|_1=r$. Given this and Assumption \ref{ass:bounded-gradients} we have $|\langle \mb{v}_i, \nabla \widehat{f_i}(\mb{x}) - \nabla\widehat{F}(\mb{x}) \rangle|\leq r\gamma$. 
    Moreover, it is worth noting 
    \[
    \sum_{i=1}^N p_i\langle \mb{v}_i^{(t)}, \nabla \widehat{f_i}(\mb{x}^{(t)}) \rangle - \langle \Delta^{(t)} , \nabla\widehat{F}(\mb{x}^{(t)}) \rangle \geq 0
    \]
    This is because
    \begin{align*}
        &\sum_{i=1}^N p_i\langle \mb{v}_i^{(t)}, \nabla \widehat{f_i}(\mb{x}^{(t)}) \rangle - \langle \Delta^{(t)} , \nabla\widehat{F}(\mb{x}^{(t)}) \rangle = \sum_{i=1}^N p_i\langle \mb{v}_i^{(t)}, \nabla \widehat{f_i}(\mb{x}^{(t)}) \rangle - \sum_{i=1}^N p_i \langle \Delta^{(t)} , \nabla\widehat{f}_i(\mb{x}^{(t)}) \rangle 
        \\
        &= \sum_{i=1}^N p_i\left(\langle \mb{v}_i^{(t)}, \nabla \widehat{f_i}(\mb{x}^{(t)}) \rangle -  \langle \Delta^{(t)} , \nabla\widehat{f}_i(\mb{x}^{(t)}) \rangle\right) \geq 0
    \end{align*}
    where in the last inequality by definition of $\mb{v}_i^{(t)}$ we have $\langle \mb{v}_i^{(t)}, \nabla \widehat{f_i}(\mb{x}^{(t)}) \rangle -  \langle \Delta^{(t)} , \nabla\widehat{f}_i(\mb{x}^{(t)}) \rangle \geq 0$.
    \end{proof}

    Followed by Lemma~\ref{lem:bound-heterogeneity-2} and equation \ref{eq:delta}, we obtain
    \begin{align}
        \widehat{F}(\mb{x}^{(t+1)})- \widehat{F}(\bx^{(t)})  
        & \geq \eta\left(\texttt{OPT} - \widehat{F}(\mb{x}^{(t)}) - r\gamma_t \right) -\frac{\eta^2\cdot m_F r^2}{2}
    \end{align}
    Now, by changing signs and adding $\texttt{OPT}$ to both sides, we get
    \begin{align}
        \texttt{OPT}-\widehat{F}(\mb{x}^{(t+1)}) & \leq (1-\eta)\left( \texttt{OPT}-\widehat{F}(\mb{x}^{(t)})\right)+ \eta r\gamma_t +\frac{\eta^2\cdot m_F r^2}{2}
    \end{align}
    Applying the same inequality inductively gives
    \begin{align}
        {\texttt{OPT}}-\widehat{F}(\mb{x}^{(t+1)}) & \leq (1-\eta)^{t+1}\left( {\texttt{OPT}}-\widehat{F}(\mb{0})\right)+ \eta r\left(\sum_{t=0}^{T-1}(1-\eta)^{T-t-1}\gamma_t\right) +\frac{(\eta^2\cdot m_F r^2)(\sum_{t=0}^{T-1}(1-\eta)^{T-t-1})}{2}\\
        &\leq (1-\eta)^{t+1}\left( {\texttt{OPT}}-\widehat{F}(\mb{0})\right)+ \eta r\left(\sum_{t=0}^{T-1}\gamma_t\right) +\frac{T\eta^2\cdot m_F r^2}{2}\tag{$(1-\eta)< 1$}
        % \\
        % &\leq (1-\eta)^{t+1}\left( \tilde{\texttt{OPT}}-\widehat{F}(\mb{0})\right)+ \eta(1-r\eta)r\gamma\frac{1-(1-\eta)^{t+1}}{\eta}
    \end{align}
    Hence for $\mb{x}^{(T)}$ we have
    \begin{align}
        \left( 1- (1-\eta)^T\right){\texttt{OPT}} \leq \widehat{F}(\mb{x}^{(T)}) + \eta r\left(\sum_{t=0}^{T-1}\gamma_t\right) +\frac{T\eta^2\cdot m_F r^2}{2}
    \end{align}
    Finally, setting $\eta=1/T$ yields
    \begin{align}
        \left(1- \frac{1}{\mathrm{e}}\right)\texttt{OPT} &\leq \widehat{F}(\mb{x}^{(T)}) + \frac{r\left(\sum_{t=0}^{T-1}\gamma_t\right)}{T} +\frac{m_F r^2}{2T}\\
        &= \widehat{F}(\mb{x}^{(T)}) + \frac{rD}{T} +\frac{m_F r^2}{2T}
    \end{align}
    for $\sum_{t=0}^{T-1}\gamma_t=D$. Note that for any $\varepsilon\geq 0$ by setting $T=\max\{\frac{2rD}{\varepsilon}, \frac{m_F r^2}{\varepsilon}\}$ we obtain
    \begin{align}
        \left(1- \frac{1}{\mathrm{e}}\right)\texttt{OPT} \leq \widehat{F}(\mb{x}^{(T)}) + \varepsilon
    \end{align}
\end{proof}

\subsection{Upper Bound for $\gamma_t$ in \Cref{ass:bounded-gradients}}

Inequality \eqref{eq:bound-gamma} in \Cref{ass:bounded-gradients} can be derived using monotonicity and submodularity as follows. For simplicity we drop $t$.
\begin{align*}
     |\nabla \widehat{f}_i(\mathbf{x}) - \nabla \widehat{F}(\mathbf{x})|_\infty 
     & = \max_{e\in E} |[\nabla \widehat{f}_i(\mathbf{x})]_e - [\nabla \widehat{F}(\mathbf{x})]_e| \\
    & \leq \max_{e\in E} [\nabla \widehat{f}_i(\mathbf{x})]_e + [\nabla \widehat{F}(\mathbf{x})]_e \tag{by monotonicity and submodularity: $[\nabla\widehat{F}(\mathbf{x})]_e, [\nabla \widehat{f}_i(\mathbf{x})]_e\geq 0$}\\
    &\leq \max_{e\in E} \max_{i} [\nabla \widehat{f}_i(\mathbf{x})]_e + [\nabla \widehat{f}_i(\mathbf{x})]_e \tag{using $F=\sum_{i} p_i f_i$ and $\sum_i p_i =1$} \\
    & = 2 \max_{e\in E} \max_{i} [\nabla \widehat{f}_i(\mathbf{x})]_e\\
    & = 2 \max_{i} \max_{e\in E} [\nabla \widehat{f}_i(\mathbf{x})]_e \\
    & = 2 \max_{i} m_{f_i}
\end{align*}

%% file: Alg1_Omitted_Proofs.tex
\section{Proof of Convergence for FedCG (\Cref{thm:convergence-Scheme2-Alg1})}

In order to prove the approximation guarantee of \Cref{thm:convergence-Scheme2-Alg1}, several steps are taken. First we show that our client sampling scheme is unbiased and furthermore we provide an upper bound for the variance.  
\subsection{Bounding the Variance}
\paragraph{Unbiased client selection.} Here we prove that our client selection is unbiased (\cref{lem:unbiased-sampling}). 
\begin{proof}[Proof of \Cref{lem:unbiased-sampling}]
     Recall that $\langle \Delta^{(t)},\nabla\widehat{F}(\bx^{(t)}) \rangle = \langle \frac{1}{|A^{(t)}|} \sum_{i\in A^{(t)}}\mb{v}_i^{(t)},\nabla\widehat{F}(\bx^{(t)}) \rangle$ where $A^{(t)}$ is a set of size $K$. 
    \allowdisplaybreaks
    \begin{align}
        &\E_{A^{(t)}}\left[\langle \Delta^{(t)},\nabla\widehat{F}(\bx^{(t)}) \rangle \right]  
        = \E_{A^{(t)}}\left[\langle \frac{1}{|A^{(t)}|}\sum_{i\in A^{(t)}} \mb{v}_i^{(t)},\nabla\widehat{F}(\bx^{(t)}) \rangle \right] \\
        &= \frac{1}{K}\E_{A^{(t)}}\left[\langle \sum_{i\in A^{(t)}} \mb{v}_i^{(t)},\nabla\widehat{F}(\bx^{(t)}) \rangle \right]
        = \frac{1}{K}\left[ \sum_{i\in A^{(t)}} \E_{A^{(t)}} \langle \mb{v}_i^{(t)},\nabla\widehat{F}(\bx^{(t)}) \rangle \right]\\
        & = \frac{1}{K}\left[ K \E_{A^{(t)}} \langle \mb{v}_i^{(t)},\nabla\widehat{F}(\bx^{(t)}) \rangle \right] \tag{for an arbitrary $i\in A^{(t)}$}\\
         & = \E_{A^{(t)}} \left[\langle \mb{v}_i^{(t)},\nabla\widehat{F}(\bx^{(t)}) \rangle  \right]
        =  \sum_{i=1}^N p_i \langle\mb{v}_i^{(t)}, \nabla\widehat{F}(\mb{x}^{(t)}) \rangle 
         =  \langle \sum_{i=1}^N p_i \mb{v}_i^{(t)}, \nabla\widehat{F}(\mb{x}^{(t)}) \rangle\label{eq:expectation_delta-scheme2}
    \end{align} 
    \end{proof}  

\paragraph{Bounding the variance.} Bounding the variance is at the core of our proof of convergence. Recall the definition of $\gamma_t$ in \Cref{ass:bounded-gradients} which plays a pivotal role in bounding the variance. 
    \begin{proof}[Proof of \Cref{lem:bounded-variance}]
    Recall $\Delta^{(t)}= \frac{1}{|A^{(t)}|}\sum_{i\in A^{(t)}}\mb{v}_i^{(t)}$ and $|A^{(t)}|=K$ and let $\bar{\mb{v}}^{(t)}=\sum_{i=1}^N p_i \mb{v}_i^{(t)}$. By \Cref{lem:unbiased-sampling} we have that $\E_{A^{(t)}}\left[\langle \Delta^{(t)},\nabla\widehat{F}(\bx^{(t)}) \rangle\right]=\langle  \sum_{i=1}^N p_i \mb{v}_i^{(t)}, \nabla\widehat{F}(\mb{x}^{(t)}) \rangle$. Furthermore, our client sampling samples $K$ clients independently and with replacement. Therefore,
    \allowdisplaybreaks
    \begin{align}
         \Var\left(\langle \Delta^{(t)},\nabla\widehat{F}(\bx^{(t)}) \rangle  \right)
        & = \frac{1}{K}\E_{A^{(t)}}\left[\left(\langle \mb{s},\nabla\widehat{F}(\bx^{(t)}) \rangle  - \langle  \bar{\mb{v}}^{(t)}, \nabla\widehat{F}(\mb{x}^{(t)}) \rangle\right)^2\right]
    \end{align}
    where $\mb{s}$ corresponds to an arbitrary client in $A^{(t)}$. Note that both terms $\langle \mb{s},\nabla\widehat{F}(\bx^{(t)}) \rangle$ and  $\langle  \bar{\mb{v}}^{(t)}, \nabla\widehat{F}(\mb{x}^{(t)}) \rangle$ are nonnegative. We provide an upper bound on the absolute value of $|\langle \mb{s},\nabla\widehat{F}(\bx^{(t)}) \rangle  - \langle  \bar{\mb{v}}^{(t)}, \nabla\widehat{F}(\mb{x}^{(t)}) \rangle|$ by considering two cases. \\
    \noindent{\textbf{Case 1.}} In the first case we have $\langle \mb{s},\nabla\widehat{F}(\bx^{(t)}) \rangle  \geq \langle  \bar{\mb{v}}^{(t)}, \nabla\widehat{F}(\mb{x}^{(t)}) \rangle$. Therefore,
    \begin{align}
        &|\langle \mb{s},\nabla\widehat{F}(\bx^{(t)}) \rangle  - \langle  \bar{\mb{v}}^{(t)}, \nabla\widehat{F}(\mb{x}^{(t)}) \rangle| \\
        & = \langle \mb{s},\nabla\widehat{F}(\bx^{(t)}) \rangle  - \langle  \bar{\mb{v}}^{(t)}, \nabla\widehat{F}(\mb{x}^{(t)}) \rangle 
         =\sum_{i=1}^N p_i \langle \mb{s},\nabla \widehat{f}_i(\mb{x}^{(t)})\rangle - \langle\sum_{i=1}^N p_i \mb{v}_i^{(t)},\nabla \widehat{F}(\mb{x}^{(t)})\rangle \\
        &\leq \sum_{i=1}^N p_i \langle \mb{v}_i^{(t)},\nabla \widehat{f}_i(\mb{x}^{(t)})\rangle - \langle\sum_{i=1}^N p_i \mb{v}_i^{(t)},\nabla \widehat{F}(\mb{x}^{(t)})\rangle
        = \sum_{i=1}^N p_i\langle \mb{v}_i^{(t)}, \nabla \widehat{f_i}(\mb{x}^{(t)}) - \nabla\widehat{F}(\mb{x}^{(t)}) \rangle \\
        &\leq \sum_{i=1}^N p_i r\gamma_t = r\gamma_t
    \end{align}
    where in the last inequality we used the fact that for each $\mb{v}_i^{(t)}$ we have $\mb{v}_i^{(t)}\in \{0,1\}^n$ and $|\mb{v}_i^{(t)}|_1=r$, and \Cref{ass:bounded-gradients}  together yield $|\langle \mb{v}_i, \nabla \widehat{f_i}(\mb{x}^{(t)}) - \nabla\widehat{F}(\mb{x}^{(t)}) \rangle|\leq r\gamma_t$.

    \noindent{\textbf{Case 2.}} In the first case we have $\langle \mb{s},\nabla\widehat{F}(\bx^{(t)}) \rangle \leq \langle  \bar{\mb{v}}^{(t)}, \nabla\widehat{F}(\mb{x}^{(t)}) \rangle$. Therefore,
    \begin{align}
        &|\langle \mb{s},\nabla\widehat{F}(\bx^{(t)}) \rangle  - \langle  \bar{\mb{v}}^{(t)}, \nabla\widehat{F}(\mb{x}^{(t)}) \rangle| 
        =  \langle  \bar{\mb{v}}^{(t)}, \nabla\widehat{F}(\mb{x}^{(t)}) \rangle - \langle \mb{s},\nabla\widehat{F}(\bx^{(t)}) \rangle\\
        & =\langle  \bar{\mb{v}}^{(t)}, \nabla\widehat{F}(\mb{x}^{(t)}) \rangle - \sum_{i=1}^N p_i\langle \mb{v}_i^{(t)}, \nabla \widehat{f_i}(\mb{x}^{(t)})\rangle - \langle \mb{s},\nabla\widehat{F}(\bx^{(t)}) \rangle + \sum_{i=1}^N p_i\langle \mb{v}_i^{(t)}, \nabla \widehat{f_i}(\mb{x}^{(t)}) \rangle\\
        &= \sum_{i=1}^N p_i \langle \mb{v}_i^{(t)},\nabla \widehat{F}(\mb{x}^{(t)})\rangle - \sum_{i=1}^N p_i\langle \mb{v}_i^{(t)}, \nabla \widehat{f_i}(\mb{x}^{(t)})\rangle + \left( \sum_{i=1}^N p_i\langle \mb{v}_i^{(t)}, \nabla \widehat{f_i}(\mb{x}^{(t)}) \rangle - \langle \mb{s},\nabla\widehat{F}(\bx^{(t)}) \rangle \right)\\
        &\leq \sum_{i=1}^N p_i \left|\langle \mb{v}_i^{(t)}, \nabla\widehat{F}(\mb{x}^{(t)})-\nabla\widehat{f_i}(\mb{x}^{(t)}) \rangle \right| +   \sum_{i=1}^N p_i \left( \langle \mb{v}_i^{(t)}, \nabla \widehat{f_i}(\mb{x}^{(t)}) \rangle - \langle \mb{s},\nabla\widehat{f}_i(\bx^{(t)}) \rangle \right) \\
        &\leq  r\gamma + \sum_{i=1}^N p_i \left( \langle \mb{v}_i^{(t)}, \nabla \widehat{f_i}(\mb{x}^{(t)}) \rangle - \langle \mb{s},\nabla\widehat{f}_i(\bx^{(t)}) \rangle \right) \label{eq:case2-variance-last}
    \end{align}
    where in the last inequality we used the same argument as the \textbf{Case 1}. Now it is left to provide an upper bound for $\langle \mb{v}_i^{(t)}, \nabla \widehat{f_i}(\mb{x}) \rangle - \langle \mb{s},\nabla\widehat{f}_i(\bx^{(t)}) \rangle$, which by definition is nonnegative.

    \begin{align}
        &\langle \mb{v}_i^{(t)}, \nabla \widehat{f_i}(\mb{x}^{(t)}) \rangle - \langle \mb{s},\nabla\widehat{f}_i(\bx^{(t)}) \rangle \\
        &= \langle \mb{v}_i^{(t)}, \nabla \widehat{f_i}(\mb{x}^{(t))} \rangle - \langle \mb{v}_i^{(t)}, \nabla \widehat{f_s}(\mb{x}^{(t)}) \rangle +\langle \mb{v}_i^{(t)}, \nabla \widehat{f_s}(\mb{x}^{(t)}) \rangle- \langle \mb{s},\nabla\widehat{f}_i(\bx^{(t)}) \rangle \\
        &= \langle \mb{v}_i^{(t)}, \nabla \widehat{f_i}(\mb{x}^{(t)}) - \nabla \widehat{f_s}(\mb{x}^{(t)})\rangle + \langle \mb{v}_i^{(t)}, \nabla \widehat{f_s}(\mb{x}^{(t)}) \rangle- \langle \mb{s},\nabla\widehat{f}_i(\bx^{(t)}) \rangle\\
        &\leq \langle \mb{v}_i^{(t)}, \nabla \widehat{f_i}(\mb{x}^{(t)}) - \nabla \widehat{f_s}(\mb{x}^{(t)})\rangle + \langle \mb{s}, \nabla \widehat{f_s}(\mb{x}^{(t)}) \rangle- \langle \mb{s},\nabla\widehat{f}_i(\bx^{(t)}) \rangle\\
        &  = \langle \mb{v}_i^{(t)}, \nabla \widehat{f_i}(\mb{x}^{(t)}) - \nabla \widehat{f_s}(\mb{x}^{(t)})\rangle + \langle \mb{s}, \nabla \widehat{f_s}(\mb{x}^{(t)})-\nabla\widehat{f}_i(\bx^{(t)}) \rangle
    \end{align}
    By \Cref{ass:bounded-gradients} we know that $|\nabla \widehat{f_s}(\mb{x}^{(t)}) - \nabla \widehat{F}(\mb{x}^{(t)})|_{\infty}\leq \gamma_t$ and $|\nabla \widehat{f_i}(\mb{x}^{(t)}) - \nabla \widehat{F}(\mb{x}^{(t)})|_{\infty}\leq \gamma_t$. Therefore, $|\nabla \widehat{f_s}(\mb{x}^{(t)}) - \nabla \widehat{f_i}(\mb{x}^{(t)})|_{\infty}\leq 2\gamma_t$. Knowing that both $\mb{v}_i^{(t)},\mb{s}\in \{0,1\}^n$ and $|\mb{v}_i^{(t)}|_1=|\mb{s}|_1=r$ yields
    \begin{align}
    &\langle \mb{v}_i^{(t)}, \nabla \widehat{f_i}(\mb{x}^{(t)}) \rangle - \langle \mb{s},\nabla\widehat{f}_i(\bx^{(t)}) \rangle\\
    &\leq \langle \mb{v}_i^{(t)}, \nabla \widehat{f_i}(\mb{x}^{(t)}) - \nabla \widehat{f_s}(\mb{x}^{(t)})\rangle + \langle \mb{s}, \nabla \widehat{f_s}(\mb{x}^{(t)})-\nabla\widehat{f}_i(\bx^{(t)}) \rangle \leq 4r\gamma_t
    \end{align}

    Putting together the above inequality and the inequality in equation~\ref{eq:case2-variance-last} we obtain the following upper bound for \textbf{Case 2}
    \begin{align*}
        |\langle \mb{s},\nabla\widehat{F}(\bx^{(t)}) \rangle  - \langle  \bar{\mb{v}}^{(t)}, \nabla\widehat{F}(\mb{x}^{(t)}) \rangle| \leq 5r\gamma_t.
    \end{align*}

    Provided the upper bounds in both cases we have
    \begin{align}
        \Var\left(\langle \Delta^{(t)},\nabla\widehat{F}(\bx^{(t)}) \rangle  \right)
        & = \frac{1}{K}\E_{A^{(t)}}\left[\left(\langle \mb{s},\nabla\widehat{F}(\bx^{(t)}) \rangle  - \langle  \bar{\mb{v}}^{(t)}, \nabla\widehat{F}(\mb{x}^{(t)}) \rangle\right)^2\right] \leq \frac{36r^2\gamma_t^2}{K}.
    \end{align}
\end{proof}

\subsection{Putting Everything Together: Proof of \Cref{thm:convergence-Scheme2-Alg1}}
We now are ready to prove \Cref{thm:convergence-Scheme2-Alg1}. 

\begin{proof}[Proof of \Cref{thm:convergence-Scheme2-Alg1}]
    Recall that $\bx^{(t+1)}\gets \bx^{(t)} + \eta\Delta^{(t)}$. Similar to the proof of \Cref{thm:full-participation} equation~\eqref{eq:delta} we derive that 
    \begin{align}
        \widehat{F}(\mb{x}^{(t+1)})
         &\geq \widehat{F}(\bx^{(t)}) + \eta\langle \Delta^{(t)},\nabla\widehat{F}(\bx^{(t)}) \rangle - \frac{\eta^2\cdot m_F r^2}{2}
    \end{align}
    
    Given \Cref{lem:unbiased-sampling,lem:bounded-variance} and using Chebyshev's inequality, over the random choices of $A^{(t)}$ and for every $\alpha>0$ we obtain
    \begin{align}
        \label{eq:Chebyshev}
        &\P\left[\left| \langle \Delta^{(t)},\nabla\widehat{F}(\bx^{(t)}) \rangle -  \langle \sum_{i=1}^N p_i \mb{v}_i^{(t)}, \nabla\widehat{F}(\mb{x}^{(t)}) \rangle \right| \leq \frac{6r\gamma_t/\sqrt{K}}{\alpha} \right] \\
        &\geq \P\left[\left| \langle \Delta^{(t)},\nabla\widehat{F}(\bx^{(t)}) \rangle -  \E_{A^{(t)}}\left[\langle \Delta^{(t)},\nabla\widehat{F}(\bx^{(t)}) \rangle \right] \right| \leq \frac{\sqrt{\Var\left( \langle \Delta^{(t)},\nabla\widehat{F}(\bx^{(t)}) \rangle  \right)}}{\alpha} \right] \geq 1-\alpha^2
    \end{align}
    
    Given this, with probability at least $1-\alpha^2$ and in the worst case it holds that 
    \begin{align}
        \widehat{F}(\mb{x}^{(t+1)})
         &\geq \widehat{F}(\bx^{(t)}) + \eta\langle \Delta^{(t)},\nabla\widehat{F}(\bx^{(t)}) \rangle - \frac{\eta^2\cdot m_F r^2}{2}\\
         &\geq \widehat{F}(\bx^{(t)}) + \eta\left( \langle \sum_{i=1}^N p_i\mb{v}_i^{(t)}, \nabla\widehat{F}(\mb{x}^{(t)}) \rangle - \frac{6r\gamma_t/\sqrt{K}}{\alpha}\right) -\frac{\eta^2\cdot m_F r^2}{2} \\
         &\geq \widehat{F}(\bx^{(t)}) + \eta\left(  \sum_{i=1}^N p_i\langle\mb{v}_i^{(t)}, \nabla\widehat{f_i}(\mb{x}^{(t)}) \rangle - r\gamma_t-\frac{6r\gamma_t/\sqrt{K}}{\alpha}\right)-\frac{\eta^2\cdot m_F r^2}{2} \tag{by Lemma~\ref{lem:bound-heterogeneity}} \\
        & \geq \widehat{F}(\bx^{(t)}) + \eta\left( \texttt{OPT} - \widehat{F}(\mb{x}^{(t)}) \right) - \eta \left( r\gamma_t+\frac{6r\gamma_t/\sqrt{K}}{\alpha}\right) -\frac{\eta^2\cdot m_F r^2}{2}\tag{by Corollary~\ref{cor:opt-F}}
    \end{align}
    Now, by changing signs and adding ${\texttt{OPT}}$ to both sides, we get
    \begin{align}
        {\texttt{OPT}}-\widehat{F}(\mb{x}^{(t+1)}) & \leq (1-\eta)\left( {\texttt{OPT}}-\widehat{F}(\mb{x}^{(t)})\right)+ \eta \left( r\gamma_t+\frac{6r\gamma_t/\sqrt{K}}{\alpha}\right) +\frac{\eta^2\cdot m_F r^2}{2}
    \end{align}
    Applying the same inequality inductively gives
    \begin{align}
        {\texttt{OPT}}-\widehat{F}(\mb{x}^{(t+1)}) &\leq (1-\eta)^{t+1}\left({\texttt{OPT}}-\widehat{F}(\mb{0}) \right)  + \eta\left( r\sum_{t=0}^{T-1}(1-\eta)^{T-t-1}\gamma_t+\frac{6r\sum_{t=0}^{T-1}(1-\eta)^{T-t-1}\gamma_t}{\alpha\sqrt{K}}\right)\\
        &+\frac{\sum_{t=0}^{T-1}(1-\eta)^{T-t-1}\eta^2\cdot m_F r^2}{2}\\
        &\leq (1-\eta)^{t+1} {\texttt{OPT}}  + \eta\left( r\sum_{t=0}^{T-1}\gamma_t+\frac{6r\sum_{t=0}^{T-1}\gamma_t}{\alpha\sqrt{K}}\right)+\frac{T\eta^2\cdot m_F r^2}{2} 
    \end{align}
    Taking the union bound over $T$ steps and $\alpha=\sqrt{\frac{\delta}{T}}$, with probability at least $1-T\cdot\alpha^2=1-\delta$, we get 
    \begin{align}
        \left( 1- (1-\eta)^T\right){\texttt{OPT}} \leq \widehat{F}(\mb{x}^{(T)}) + \eta\left( r\sum_{t=0}^{T-1}\gamma_t+\frac{6r\sum_{t=0}^{T-1}\gamma_t}{\sqrt{K\delta/T}}\right)+\frac{T\eta^2\cdot m_F r^2}{2}
    \end{align}
    Setting $\eta=1/T$ yields
    \begin{align}
        \left( 1- 1/\mathrm{e}\right)\texttt{OPT} &\leq \widehat{F}(\mb{x}^{(T)})+ \frac{1}{T}\left( r\sum_{t=0}^{T-1}\gamma_t+\frac{6r\sum_{t=0}^{T-1}\gamma_t}{\sqrt{K\delta/T}}\right)+\frac{m_F r^2}{2T}\\
        &=\widehat{F}(\mb{x}^{(T)})+\frac{rD}{T}+\frac{6rD}{\sqrt{KT\delta}}+\frac{m_F r^2}{2T}
    \end{align}
    for $\sum_{t=0}^{T-1}\gamma_t=D$. Note that for any $\varepsilon\geq 0$ by setting $T=\max\{\frac{3rD}{\varepsilon}, \frac{3m_F r^2}{2\varepsilon}, \frac{324r^2D^2}{K\delta\varepsilon}\}$ we obtain
    \begin{align}
        \left(1- \frac{1}{\mathrm{e}}\right)\texttt{OPT} \leq \widehat{F}(\mb{x}^{(T)}) + \varepsilon
    \end{align}
\end{proof}
% \begin{proof}[Proof of Theorem~\ref{thm:convergence-Scheme2-Alg1}]
% The proof is identical to the proof of Theorem~\ref{thm:convergence-Scheme1-Alg1} except we should use $\Var\left(\eta \langle \Delta^{(t)},\nabla\widehat{F}(\bx^{(t)}) \rangle  \right) \leq \frac{\eta^2\gamma^2r}{K}$.    
% \end{proof}

%% file: Alg2_Omitted_Proofs.tex
\section{Proof of Convergence for \texttt{\textsc{Fed}CG+} (\Cref{thm:convergence-Alg2-Scheme2})}

\label{sec:convergence-FEDCG+}
\subsection{Bounding the variance}

\paragraph{Unbiased client selection.} Here we prove that our client selection is unbiased (\Cref{lem:unbiased-sampling-Alg2}). The proof is almost identical to the one for \Cref{lem:unbiased-sampling}. We present a proof here for the sake of completeness.

\begin{proof}[Proof of \Cref{lem:unbiased-sampling-Alg2}] 
    
    Recall that $\langle \widetilde{\Delta}^{(t+\tau)},\nabla\widehat{F}(\bx^{(t)}) \rangle = \langle \frac{1}{|A^{(t)}|} \sum_{i\in A^{(t)}}\widetilde{\Delta}_i^{(t+\tau)},\nabla\widehat{F}(\bx^{(t)}) \rangle$ where $A^{(t)}$ is a set of size $K$. 
    \allowdisplaybreaks
    \begin{align}
        &\E_{A^{(t)}}\left[\langle \widetilde{\Delta}^{(t+\tau)},\nabla\widehat{F}(\bx^{(t)}) \rangle \right]  
        = \E_{A^{(t)}}\left[\langle \frac{1}{|A^{(t)}|}\sum_{i\in A^{(t)}} \widetilde{\Delta}_i^{(t+\tau)},\nabla\widehat{F}(\bx^{(t)}) \rangle \right] \\ \tag{$|A^{(t)}|=K$}\\
         &= \frac{1}{K}\E_{A^{(t)}}\left[\langle \sum_{i\in A^{(t)}} \widetilde{\Delta}_i^{(t+\tau)},\nabla\widehat{F}(\bx^{(t)}) \rangle \right] = \frac{1}{K}\left[ \sum_{i\in A^{(t)}} \E_{A^{(t)}} \langle \widetilde{\Delta}_i^{(t+\tau)},\nabla\widehat{F}(\bx^{(t)}) \rangle \right]\\
        & = \frac{1}{K}\left[ K \E_{A^{(t)}} \langle \widetilde{\Delta}_i^{(t+\tau)},\nabla\widehat{F}(\bx^{(t)}) \rangle \right] \tag{for an arbitrary $i\in A^{(t)}$}\\
         & = \E_{A^{(t)}} \langle \mb{v}_i^{(t)},\nabla\widehat{F}(\bx^{(t)}) \rangle 
          =  \sum_{i=1}^N p_i \langle \widetilde{\Delta}_i^{(t+\tau)}, \nabla\widehat{F}(\mb{x}^{(t)}) \rangle 
          =  \langle \sum_{i=1}^N p_i \widetilde{\Delta}_i^{(t+\tau)}, \nabla\widehat{F}(\mb{x}^{(t)}) \rangle 
    \end{align} 
\end{proof}

\paragraph{Bounding the variance.} Bounding the variance is at the core of our proof of convergence. \Cref{ass:bounded-gradients} plays a pivotal role in bounding the variance. The main difficulty and difference between this proof and the proof of \Cref{lem:bounded-variance} is that firstly the variance is effected by the divergence caused by local steps, and secondly $\widetilde{\Delta}_{i}^{(t+\tau)}$ may not be integral vectors and it could potentially have $O(n)$ nonzero entries. 

In order to handle the divergence caused by local steps we assume Lipschitzness. First, let us derive useful inequalities using the Lipschitzness condition.\\

    \noindent{\textbf{Consequences of Lipschitzness}.} Consider $t$-th iteration and recall the definition of $L_t$ i.e., $L_t$ is such that for all $i\in[N]$ and $j\in[\tau]$ it holds
    \begin{align*}
        \left \| \nabla \widehat{f_i}(\bx_i^{(t,j)})  - \nabla \widehat{f_i}(\bx^{(t)}) \right \| \leq L_t\| \mb{x}_i^{(t,j)} - \mb{x}^{(t)} \| 
    \end{align*}
    
    %all $\widehat{f}_i$ have Lipschitz continuous gradients i.e., $\|\nabla\widehat{f}_i(\mb{x}^{(t)})-\nabla\widehat{f}_i(\mb{y})\|\leq L\|\mb{x}-\mb{y}\|$, $\forall\mb{x},\mb{y}\in\mc{P}$. 
    where $\mb{x}_i^{(t)}=\mb{x}^{(t)}$ and $\mb{x}_i^{(t,j)}$ are the local models for client $i$ at time step $t$ and $t+j$, respectively. We first bound the divergence of $\mb{x}_i^{(t,\tau)}$ from $\mb{x}^{(t)}$.
        \begin{align}
            \left\| \mb{x}_i^{(t,\tau)} - \mb{x}^{(t)} \right\| = \left\| \frac{1}{\tau} \sum_{j=0}^{\tau-1} \widetilde{\mb{v}}_i^{(t,j)} \right\| \leq \frac{1}{\tau}  \sum_{j=0}^{\tau-1} \left\| \widetilde{\mb{v}}_i^{(t,j)} \right\| \leq \sqrt{r}
        \end{align}
        Since $\widetilde{\mb{v}}_i^{(t,j)}\in\mc{P}$, the same upper bound holds for every $0\leq j\leq\tau$; $\| \mb{x}_i^{(t,j)} - \mb{x}^{(t)} \|\leq  \sqrt{r}$.
         Assuming $L_t$, for all $0\leq  j\leq \tau$,
        \begin{align}
        \label{eq:Lipschitz}
            \left \| \nabla \widehat{f_i}(\bx_i^{(t,j)})  - \nabla \widehat{f_i}(\bx^{(t)}) \right \| \leq L\| \mb{x}_i^{(t,j)} - \mb{x}^{(t)} \| \leq L_t \sqrt{r}
        \end{align}
        Note $\zeta_i^{(t,j)}$ are sampled according to $\bx_i^{(t,j)}$, and $\nabla\widetilde{f_i}(\mb{x}_i^{(t,j)},\zeta_{i}^{(t,j)})$ estimates $\nabla\widehat{f_i}(\mb{x}_i^{(t,j)})$ within factor $\sigma$ i.e., $\|\nabla\widehat{f_i}(\mb{x}_i^{(t,j)}) -\nabla\widetilde{f_i}(\mb{x}_i^{(t,j)},\zeta_{i}^{(t,j)})\|\leq \sigma$. Hence, given this estimation and equation \eqref{eq:Lipschitz}, for every $0\leq j\leq\tau$ it holds that
        \begin{align}
        \label{eq:Lipschitz-estimate}
            \left \| \nabla\widetilde{f_i}(\mb{x}_i^{(t,j)},\zeta_{i}^{(t,j)})  - \nabla \widehat{f_i}(\bx^{(t)}) \right \|  \leq \sigma+L_t  \sqrt{r}
        \end{align}

\begin{proof}[Proof of \Cref{lem:bounded-variance-Alg2}]
    By \Cref{lem:unbiased-sampling-Alg2} we know: 
    \begin{align*}
        \E_{A^{(t)}}\left[\langle \widetilde{\Delta}^{(t+\tau)},\nabla\widehat{F}(\bx^{(t)}) \rangle \right] =  \langle \sum_{i=1}^N p_i \widetilde{\Delta}_i^{(t+\tau)}, \nabla\widehat{F}(\mb{x}^{(t)}) \rangle
    \end{align*}
     where $\widetilde{\Delta}^{(t+\tau)}= \frac{1}{|A^{(t)}|}\sum_{i\in A^{(t)}}\widetilde{\Delta}_i^{(t+\tau)}$.  Define $\overline{\Delta}^{(t+\tau)}=\sum_{i=1}^N p_i \widetilde{\Delta}_i^{(t+\tau)}$. Each client is selected to be in set $A^{(t)}$ independently and with replacement. Therefore,
    \begin{align}
        &\Var(\langle \widetilde{\Delta}^{(t+\tau)},\nabla\widehat{F}(\bx^{(t)}) \rangle)=\frac{1}{K}\E_{A^{(t)}}\left[\left(\langle \widetilde{\Delta}_{s}^{(t+\tau)},\nabla\widehat{F}(\bx^{(t)}) \rangle  -  \langle \overline{\Delta}^{(t+\tau)}, \nabla\widehat{F}(\mb{x}^{(t)}) \rangle\right)^2\right]
    \end{align}
    where $\widetilde{\Delta}_{s}^{(t+\tau)}=\bx_s^{(t,\tau)}-\bx_s^{(t,0)}$ corresponds to an arbitrary client in $A^{(t)}$. Note that both terms $\langle \widetilde{\Delta}_{s}^{(t+\tau)},\nabla\widehat{F}(\bx^{(t)}) \rangle$ and  $\langle \overline{\Delta}^{(t+\tau)}, \nabla\widehat{F}(\mb{x}^{(t)}) \rangle$ are nonnegative. We provide an upper bound on the absolute value of $\left| \langle \widetilde{\Delta}_{s}^{(t+\tau)},\nabla\widehat{F}(\bx^{(t)}) \rangle  -  \langle \overline{\Delta}^{(t+\tau)}, \nabla\widehat{F}(\mb{x}^{(t)}) \rangle\right|$ by considering two cases. 
    
    \noindent{\textbf{Case 1.}} In the first case we have $\langle \widetilde{\Delta}_{s}^{(t+\tau)},\nabla\widehat{F}(\bx^{(t)}) \rangle  \geq  \langle \overline{\Delta}^{(t+\tau)}, \nabla\widehat{F}(\mb{x}^{(t)}) \rangle$. Therefore,
    \begin{align}
        & \left|\langle \widetilde{\Delta}_{s}^{(t+\tau)},\nabla\widehat{F}(\bx^{(t)}) \rangle  -  \langle \overline{\Delta}^{(t+\tau)}, \nabla\widehat{F}(\mb{x}^{(t)}) \rangle \right| \\
       & = \langle \widetilde{\Delta}_{s}^{(t+\tau)},\nabla\widehat{F}(\bx^{(t)}) \rangle  -  \langle \overline{\Delta}^{(t+\tau)}, \nabla\widehat{F}(\mb{x}^{(t)}) \rangle \\
        & =\sum_{i=1}^N p_i \langle \widetilde{\Delta}_{s}^{(t+\tau)},\nabla \widehat{f}_i(\mb{x}^{(t)})\rangle - \langle\sum_{i=1}^N p_i \widetilde{\Delta}_i^{(t+\tau)},\nabla \widehat{F}(\mb{x}^{(t)})\rangle 
        \end{align}

        Observe that $\widetilde{\Delta}_{s}^{(t+\tau)}=\frac{1}{\tau}\sum_{j=0}^{\tau-1} \widetilde{\mb{v}}_s^{(t,j)}$ with $\widetilde{\mb{v}}_s^{(t,j)}=\argmax_{\mb{v}\in P(M)} \langle \mb{v},\nabla\widetilde{f}_s(\bx_s^{(t,j)},\zeta_{s}^{(t,j)})\rangle$. 
        
        Therefore, using the Lipschitzness condition we get the following. In what follows let $\mb{d}_1=L_t\sqrt{r}\mb{1}$ and $\mb{d}_2=\sigma\mb{1}$ be  vectors of length $n$ where every components are $L_t\sqrt{r}$ and $\sigma$, respectively.
        \allowdisplaybreaks
        \begin{align}
            &\langle \widetilde{\Delta}_{s}^{(t+\tau)},\nabla \widehat{f}_i(\mb{x}_i^{(t)})\rangle = \frac{1}{\tau}\sum_{j=0}^{\tau-1} \langle \widetilde{\mb{v}}_s^{(t,j)},\nabla \widehat{f}_i(\mb{x}_i^{(t)})\rangle \\
           %& \leq \frac{1}{\tau}\sum_{j=0}^{\tau-1} \langle \widetilde{\mb{v}}_s^{(t,j)},\nabla\widehat{f}_i(\mb{x}_i^{(t,j)})+\mb{d}_1\rangle\tag{by equation \eqref{eq:Lipschitz}}\\
            &\leq \frac{1}{\tau}\sum_{j=0}^{\tau-1} \langle \widetilde{\mb{v}}_s^{(t,j)},\nabla\widetilde{f}_i(\mb{x}_i^{(t,j)})+\mb{d}_1+\mb{d}_2\rangle \tag{by equation~\eqref{eq:Lipschitz-estimate}}\\
            & \leq \frac{1}{\tau}\sum_{j=0}^{\tau-1} \langle \widetilde{\mb{v}}_i^{(t,j)},\nabla\widetilde{f}_i(\mb{x}_i^{(t,j)})+\mb{d}_1+\mb{d}_2\rangle\tag{by definition of $\widetilde{\mb{v}}_i^{(t,j)}$}\\
            & \leq \frac{1}{\tau}\sum_{j=0}^{\tau-1} \langle \widetilde{\mb{v}}_i^{(t,j)},\nabla\widehat{f}_i(\bx_i^{(t)})\rangle + \frac{1}{\tau}\sum_{j=0}^{\tau-1} \langle \widetilde{\mb{v}}_i^{(t,j)},\mb{d}_1+\mb{d}_2\rangle\\
            &\leq \frac{1}{\tau}\sum_{j=0}^{\tau-1} \langle \widetilde{\mb{v}}_i^{(t,j)},\nabla\widehat{f}_i(\bx_i^{(t)})\rangle + (\sigma r+L_tr^{1.5}) \\
            & = \langle \widetilde{\Delta}_{i}^{(t+\tau)},\nabla\widehat{f}_i(\bx_i^{(t)})\rangle + (\sigma r+L_tr^{1.5})  \label{eq:interlacing}\\
            & = \langle \widetilde{\Delta}_{i}^{(t+\tau)},\nabla\widehat{f}_i(\bx^{(t)})\rangle + (\sigma r+L_tr^{1.5})  \tag{$\bx_i^{(t)}=\bx^{(t)}$}
        \end{align}
        Therefore, for \textbf{Case 1} we get
        \begin{align}
        &\sum_{i=1}^N p_i \langle \widetilde{\Delta}_{s}^{(t+\tau)},\nabla \widehat{f}_i(\mb{x}^{(t)})\rangle - \langle\sum_{i=1}^N p_i \widetilde{\Delta}_i^{(t+\tau)},\nabla \widehat{F}(\mb{x}^{(t)})\rangle \\
        & \leq \sum_{i=1}^N p_i\langle \widetilde{\Delta}_{i}^{(t+\tau)},\nabla\widehat{f}_i(\bx^{(t)})\rangle  - \langle\sum_{i=1}^N p_i \widetilde{\Delta}_i^{(t+\tau)},\nabla \widehat{F}(\mb{x}^{(t)})\rangle + (\sigma r+L_tr^{1.5}) \\
        &= \sum_{i=1}^N p_i\langle \widetilde{\Delta}_{i}^{(t+\tau)}, \nabla \widehat{f_i}(\mb{x}) - \nabla\widehat{F}(\mb{x}) \rangle + (\sigma r+L_tr^{1.5}) \\
        &= \sum_{i=1}^N p_i \left(\frac{1}{\tau}\sum_{j=0}^{\tau-1} \langle \widetilde{\mb{v}}_i^{(t,j)}, \nabla \widehat{f_i}(\mb{x}) - \nabla\widehat{F}(\mb{x}) \rangle \right) + (\sigma r+L_tr^{1.5}) \\
        & \leq \sum_{i=1}^N p_i \left(\frac{1}{\tau}\sum_{j=0}^{\tau-1} r\gamma_t \right) + (\sigma r+L_tr^{1.5})  \tag{\Cref{ass:bounded-gradients}}\\
        &\leq r\gamma_t + (\sigma r+L_tr^{1.5}) 
    \end{align}
    Note that in the above we used the fact that for each $\widetilde{\mb{v}}_i^{(t,j)}$ we have $\widetilde{\mb{v}}_i^{(t,j)}\in \{0,1\}^n$ and $|\widetilde{\mb{v}}_i^{(t,j)}|_1=r$.

    \noindent{\textbf{Case 2}.} In this case we have $\langle \widetilde{\Delta}_{s}^{(t+\tau)},\nabla\widehat{F}(\bx^{(t)}) \rangle  \leq  \langle \overline{\Delta}^{(t+\tau)}, \nabla\widehat{F}(\mb{x}^{(t)}) \rangle$. Therefore,
    \begin{align}
        &|\langle \widetilde{\Delta}_{s}^{(t+\tau)},\nabla\widehat{F}(\bx^{(t)}) \rangle  -  \langle \overline{\Delta}^{(t+\tau)}, \nabla\widehat{F}(\mb{x}^{(t)}) \rangle| \\
        & =  \langle \overline{\Delta}^{(t+\tau)}, \nabla\widehat{F}(\mb{x}^{(t)}) \rangle - \langle \widetilde{\Delta}_{s}^{(t+\tau)},\nabla\widehat{F}(\bx^{(t)}) \rangle\\
        & =\langle \overline{\Delta}^{(t+\tau)}, \nabla\widehat{F}(\mb{x}^{(t)}) \rangle - \sum_{i=1}^N p_i\langle \widetilde{\Delta}_{i}^{(t+\tau)},\nabla\widehat{f}_i(\bx^{(t)})\rangle - \langle \widetilde{\Delta}_{s}^{(t+\tau)},\nabla\widehat{F}(\bx^{(t)}) \rangle \\
        &+ \sum_{i=1}^N p_i\langle \widetilde{\Delta}_{i}^{(t+\tau)},\nabla\widehat{f}_i(\bx^{(t)})\rangle\\
        &\leq \sum_{i=1}^N p_i \left|\langle \widetilde{\Delta}_{i}^{(t+\tau)}, \nabla\widehat{F}(\mb{x}^{(t)})-\nabla\widehat{f_i}(\mb{x}^{(t)}) \rangle \right| \\
        &+   \sum_{i=1}^N p_i \left( \langle \widetilde{\Delta}_{i}^{(t+\tau)}, \nabla \widehat{f_i}(\mb{x}^{(t)}) \rangle - \langle \widetilde{\Delta}_{s}^{(t+\tau)},\nabla\widehat{f}_i(\bx^{(t)}) \rangle \right) \\
        & \leq \sum_{i=1}^N \frac{p_i}{\tau} \sum_{j=0}^{\tau-1}\left|\langle \widetilde{\mb{v}}_{i}^{(t,j)}, \nabla\widehat{F}(\mb{x}^{(t)})-\nabla\widehat{f_i}(\mb{x}^{(t)}) \rangle \right| \\
        &+   \sum_{i=1}^N p_i \left( \langle \widetilde{\Delta}_{i}^{(t+\tau)}, \nabla \widehat{f_i}(\mb{x}^{(t)}) \rangle - \langle \widetilde{\Delta}_{s}^{(t+\tau)},\nabla\widehat{f}_i(\bx^{(t)}) \rangle \right) \\
        &\leq  r\gamma_t + \sum_{i=1}^N p_i \left( \langle \widetilde{\Delta}_{i}^{(t+\tau)}, \nabla \widehat{f_i}(\mb{x}^{(t)}) \rangle - \langle \widetilde{\Delta}_{s}^{(t+\tau)},\nabla\widehat{f}_i(\bx^{(t)}) \rangle \right) \label{eq:case2-variance-last-local}
    \end{align}
    where in the last inequality we used the same argument by noting $\widetilde{\mb{v}}_i^{(t,j)}\in \{0,1\}^n$ and $|\widetilde{\mb{v}}_i^{(t,j)}|_1=r$. Now it is left to provide an upper bound for $\left|\langle \widetilde{\Delta}_{i}^{(t+\tau)}, \nabla \widehat{f_i}(\mb{x}^{(t)}) \rangle - \langle \widetilde{\Delta}_{s}^{(t+\tau)},\nabla\widehat{f}_i(\bx^{(t)}) \rangle \right|$.

    \begin{align}
        & \left| \langle \widetilde{\Delta}_{i}^{(t+\tau)}, \nabla \widehat{f_i}(\mb{x}^{(t)}) \rangle - \langle \widetilde{\Delta}_{s}^{(t+\tau)},\nabla\widehat{f}_i(\bx^{(t)}) \rangle \right| \\
        &= \left|\langle \widetilde{\Delta}_{i}^{(t+\tau)}, \nabla \widehat{f_i}(\mb{x}^{(t)}) \rangle- \langle \widetilde{\Delta}_{i}^{(t+\tau)}, \nabla \widehat{f_s}(\mb{x}^{(t)}) \rangle +\langle \widetilde{\Delta}_{i}^{(t+\tau)}, \nabla \widehat{f_s}(\mb{x}^{(t)}) \rangle - \langle \widetilde{\Delta}_{s}^{(t+\tau)},\nabla\widehat{f}_i(\bx^{(t)}) \rangle \right|\\
        &= \left|\langle \widetilde{\Delta}_{i}^{(t+\tau)}, \nabla \widehat{f_i}(\mb{x}^{(t)}) -  \nabla \widehat{f_s}(\mb{x}^{(t)}) \rangle +\langle \widetilde{\Delta}_{i}^{(t+\tau)}, \nabla \widehat{f_s}(\mb{x}^{(t)}) \rangle - \langle \widetilde{\Delta}_{s}^{(t+\tau)},\nabla\widehat{f}_i(\bx^{(t)}) \rangle \right|\\
        &\leq \left|\langle \widetilde{\Delta}_{i}^{(t+\tau)}, \nabla \widehat{f_i}(\mb{x}^{(t)}) -  \nabla \widehat{f_s}(\mb{x}^{(t)}) \rangle \right|+ \left| \langle \widetilde{\Delta}_{i}^{(t+\tau)}, \nabla \widehat{f_s}(\mb{x}^{(t)}) \rangle - \langle \widetilde{\Delta}_{s}^{(t+\tau)},\nabla\widehat{f}_i(\bx^{(t)}) \rangle \right|\label{eq:case2}%\\
        % &= \left|\langle \widetilde{\Delta}_{i}^{(t+\tau)}, \nabla \widehat{f_i}(\mb{x}^{(t)}) -  \nabla \widehat{f_s}(\mb{x}^{(t)}) \rangle +\langle \widetilde{\Delta}_{s}^{(t+\tau)}, \nabla \widehat{f_s}(\mb{x}^{(t)}) -\nabla\widehat{f}_i(\bx^{(t)}) \rangle \right|\\
        % & \leq \left|\langle \widetilde{\Delta}_{i}^{(t+\tau)}, \nabla \widehat{f_i}(\mb{x}^{(t)}) -  \nabla \widehat{f_s}(\mb{x}^{(t)}) \rangle \right| + \left| \langle \widetilde{\Delta}_{s}^{(t+\tau)}, \nabla \widehat{f_s}(\mb{x}^{(t)}) -\nabla\widehat{f}_i(\bx^{(t)}) \rangle \right|
    \end{align}
    By \Cref{ass:bounded-gradients} we know that $|\nabla \widehat{f_s}(\mb{x}^{(t)}) - \nabla \widehat{F}(\mb{x}^{(t)})|_{\infty}\leq \gamma_t$ and $|\nabla \widehat{f_i}(\mb{x}^{(t)}) - \nabla \widehat{F}(\mb{x}^{(t)})|_{\infty}\leq \gamma_t$. Therefore, $|\nabla \widehat{f_s}(\mb{x}^{(t)}) - \nabla \widehat{f_i}(\mb{x}^{(t)})|_{\infty}\leq 2\gamma_t$. This is used to upper bound the first term in \eqref{eq:case2}. In order to bound the second term in \eqref{eq:case2} we appeal to equation \eqref{eq:interlacing} in \textbf{Case 1}. (Note that there are two cases two consider here because of the absolute value, however the argument is similar and we argue about one case.) 
    Let $\mb{1}$ be the all one vector of length $n$, then
    \begin{align}
        &\left|\langle \widetilde{\Delta}_{i}^{(t+\tau)}, \nabla \widehat{f_i}(\mb{x}^{(t)}) -  \nabla \widehat{f_s}(\mb{x}^{(t)}) \rangle \right| + \left| \langle \widetilde{\Delta}_{i}^{(t+\tau)}, \nabla \widehat{f_s}(\mb{x}^{(t)}) \rangle - \langle \widetilde{\Delta}_{s}^{(t+\tau)},\nabla\widehat{f}_i(\bx^{(t)}) \rangle \right|\\
        &\leq \frac{1}{\tau}\sum_{j=0}^{\tau-1}\left|\langle \widetilde{\mb{v}}_{i}^{(t,j)}, \nabla \widehat{f_i}(\mb{x}^{(t)}) -  \nabla \widehat{f_s}(\mb{x}^{(t)}) \rangle \right| + \left| \langle \widetilde{\Delta}_{s}^{(t+\tau)}, \nabla \widehat{f_s}(\mb{x}^{(t)}) \rangle - \langle \widetilde{\Delta}_{s}^{(t+\tau)},\nabla\widehat{f}_i(\bx^{(t)}) \rangle \right| \\
        &+ (\sigma r+L_tr^{1.5}) \\
        &\leq \frac{1}{\tau}\sum_{j=0}^{\tau-1}\left|\langle \widetilde{\mb{v}}_{i}^{(t,j)}, \nabla \widehat{f_i}(\mb{x}^{(t)}) -  \nabla \widehat{f_s}(\mb{x}^{(t)}) \rangle \right|+  \frac{1}{\tau}\sum_{j=0}^{\tau-1} \left|\langle \widetilde{\mb{v}}_{s}^{(t,j)}, \nabla \widehat{f_s}(\mb{x}^{(t)}) -\nabla\widehat{f}_i(\bx^{(t)})  \rangle \right|\\
        &+ (\sigma r+L_tr^{1.5}) \\
        &\leq \frac{1}{\tau}\sum_{j=0}^{\tau-1}\langle \widetilde{\mb{v}}_{i}^{(t,j)}, 2\gamma_t\mb{1} \rangle +  \frac{1}{\tau}\sum_{j=0}^{\tau-1} \langle \widetilde{\mb{v}}_{s}^{(t,j)}, 2\gamma_t\mb{1} \rangle + (\sigma r+L_tr^{1.5}) \\
        &\leq 4r\gamma_t + (\sigma r+L_tr^{1.5}) 
    \end{align}
    where in the last inequality we used the fact that each $\widetilde{\mb{v}}_{i}^{(t,j)},\widetilde{\mb{v}}_{s}^{(t,j)}\in \{0,1\}^n$ and $|\widetilde{\mb{v}}_{i}^{(t,j)}|_1=|\widetilde{\mb{v}}_{s}^{(t,j)}|_1=r$.

    Putting together the above inequality and the inequality in equation~\ref{eq:case2-variance-last-local} we obtain the following upper bound for \textbf{Case 2}
    \begin{align*}
        \left| \langle \widetilde{\Delta}_{s}^{(t+\tau)},\nabla\widehat{F}(\bx^{(t)}) \rangle  -  \langle \overline{\Delta}^{(t+\tau)}, \nabla\widehat{F}(\mb{x}^{(t)}) \rangle\right| \leq 5r\gamma_t + (\sigma r+L_tr^{1.5}).
    \end{align*}
    Finally, we are at the place where we can present our upper bound for the variance
    \begin{align*}
        \Var(\langle \widetilde{\Delta}^{(t+\tau)},\nabla\widehat{F}(\bx^{(t)}) \rangle)
        &\leq \frac{1}{K}\E_{A^{(t)}}\left[\left( 6r\gamma + 2(\sigma r+Lr^{1.5}) \right)^2\right] \\
        &= \frac{1}{K}\left( 6r\gamma_t + 2(\sigma r+L_tr^{1.5}) \right)^2.\qedhere
    \end{align*}
\end{proof}

\subsection{Proof of \Cref{thm:convergence-Alg2-Scheme2}}
\paragraph{Convergence of \texttt{\textsc{Fed}CG+}.} While at the high level this proof is similar to the previous convergence proofs in Theorems~3.2, 3.5, it still requires taking care of the error caused due to the local steps. (This is an additive error term that cannot be controlled by sampling more clients at each round.) 
\begin{proof}[Proof of \Cref{thm:convergence-Alg2-Scheme2}]
    For any $t\in\mc{I}_\tau$ we analyze the difference between $\widehat{F}(\mb{x}^{(t)})$ and $\widehat{F}(\mb{x}^{(t+\tau)})$. Recall that $\bx^{(t+\tau)}\gets \bx^{(t)} + \eta\widetilde{\Delta}^{(t+\tau)}$ where $\eta$ is the server's learning rate and $\widetilde{\Delta}^{(t+\tau)}=\frac{1}{|A^{(t)}|}\sum_{i\in A^{(t)}} \widetilde{\Delta}_i^{(t+\tau)} = \frac{1}{\tau|A^{(t)}|}\sum_{i\in A^{(t)}} \sum_{j=0}^{\tau-1}\widetilde{\mb{v}}_i^{(t,j)}$. Similar to the proof of Theorem~3.2  equation~\eqref{eq:delta} we derive that 
    \begin{align}
        \widehat{F}(\mb{x}^{(t+\tau)})- \widehat{F}(\bx^{(t)})  
        & \geq \eta\langle \widetilde{\Delta}^{(t+\tau)}, \nabla\widehat{F}(\bx^{(t)} \rangle -\frac{\eta^2\cdot m_F r^2}{2}\label{eq:Taylor-expansion-Alg2}
    \end{align}

    Let $\overline{\Delta}^{(t+\tau)}= \sum_{i=1}^N p_i \widetilde{\Delta}_i^{(t+\tau)}$ and by \Cref{lem:unbiased-sampling-Alg2} we have that $\E_{A^{(t)}}\left[\eta\langle \widetilde{\Delta}^{(t+\tau)},\nabla\widehat{F}(\bx^{(t)}) \rangle \right] =  \eta\langle \overline{\Delta}^{(t+\tau)}, \nabla\widehat{F}(\mb{x}^{(t)}) \rangle$. Given Lemma~4.2 and using Chebyshev's inequality, over the random choices of $A^{(t)}$ and for every $\alpha>0$ and $\chi^2=\frac{1}{K}\left( 6r\gamma_t + 2(\sigma r+L_tr^{1.5}) \right)^2$ we obtain
    \begin{align}
        \label{eq:Chebyshev-Alg2}
        &\P\left[\left| \langle \widetilde{\Delta}^{(t+\tau)},\nabla\widehat{F}(\bx^{(t)}) \rangle -  \langle \overline{\Delta}^{(t+\tau)}, \nabla\widehat{F}(\mb{x}^{(t)}) \rangle \right| \leq \frac{\chi}{\alpha} \right] \geq 1-\alpha^2
    \end{align}
    
    Given this and \eqref{eq:Taylor-expansion-Alg2}, with probability at least $1-\alpha^2$ and in the worst case it holds that 
    \begin{align}
        \widehat{F}(\mb{x}^{(t+\tau)})
         &\geq \widehat{F}(\bx^{(t)}) + \eta\langle  \widetilde{\Delta}^{(t+\tau)}, \nabla\widehat{F}(\bx^{(t)} \rangle-\frac{\eta^2\cdot m_F r^2}{2}\\
         &\geq \widehat{F}(\bx^{(t)}) + \eta\left( \langle \overline{\Delta}^{(t+\tau)}, \nabla\widehat{F}(\mb{x}^{(t)}) \rangle - \frac{\chi}{\alpha}\right) -\frac{\eta^2\cdot m_F r^2}{2}
    \end{align}
    \begin{claim}
        \label{claim:tildeDelta-barDelta}
        Let $\gamma_t$ be as in Assumption~\ref{ass:bounded-gradients}, we have $\langle \overline{\Delta}^{(t+\tau)}, \nabla\widehat{F}(\mb{x}^{(t)}) \rangle \geq  \sum_{i=1}^N p_i\langle  \widetilde{\Delta}_i^{(t+\tau)}, \nabla\widehat{f}_i(\mb{x}^{(t)}) \rangle - r\gamma_t $.
    \end{claim}
    \begin{proof}
        Proof is similar to the proof of Lemma~\ref{lem:bound-heterogeneity}. 
        \begin{align*}
            &\sum_{i=1}^N p_i\langle  \widetilde{\Delta}_i^{(t+\tau)}, \nabla\widehat{f}_i(\mb{x}^{(t)}) \rangle-\langle \overline{\Delta}^{(t+\tau)}, \nabla\widehat{F}(\mb{x}^{(t)}) \rangle =   \sum_{i=1}^N p_i\langle  \widetilde{\Delta}_i^{(t+\tau)}, \nabla\widehat{f}_i(\mb{x}^{(t)}) - \nabla\widehat{F}(\mb{x}^{(t)}) \rangle\\
            &=\sum_{i=1}^N \sum_{j=0}^{\tau-1}\frac{p_i}{\tau}\langle  \widetilde{\mb{v}}_i^{(t,j)}, \nabla\widehat{f}_i(\mb{x}^{(t)}) - \nabla\widehat{F}(\mb{x}^{(t)}) \rangle \le \sum_{i=1}^N \sum_{j=0}^{\tau-1}\frac{p_i}{\tau}r\gamma_t=r\gamma_t \qedhere
        \end{align*}
    \end{proof}

    Therefore, 
    \begin{align}
    \label{eq:before-claim}
        \widehat{F}(\mb{x}^{(t+\tau)})
         &\geq \widehat{F}(\bx^{(t)}) + \eta (\sum_{i=1}^N p_i\langle  \widetilde{\Delta}_i^{(t+\tau)}, \nabla\widehat{f}_i(\mb{x}^{(t)}) \rangle - r\gamma_t - \frac{\chi}{\alpha})-\frac{\eta^2\cdot m_F r^2}{2}
    \end{align}
    \begin{claim}
    \label{claim:fedcg+convergence}
        $\langle  \widetilde{\Delta}_i^{(t+\tau)}, \nabla\widehat{f}_i(\mb{x}^{(t)}) \rangle \geq \langle  \mb{v}_i^{(t)}, \nabla\widehat{f}_i(\mb{x}^{(t)}) \rangle- 2(\sigma r+L_t r^{1.5})$. 
    \end{claim}
    Hence, by Claim~\ref{claim:fedcg+convergence} and equation \eqref{eq:before-claim} we have
    \begin{align}
        \widehat{F}(\mb{x}^{(t+\tau)})
         &\geq \widehat{F}(\bx^{(t)}) + \eta (\sum_{i=1}^N p_i\langle  \widetilde{\Delta}_i^{(t+\tau)}, \nabla\widehat{f}_i(\mb{x}^{(t)}) \rangle - r\gamma_t - \frac{\chi}{\alpha}-2(\sigma r+L_t r^{1.5})) - \frac{\eta^2\cdot m_F r^2}{2}
    \end{align}
    
    Now, by changing signs and adding ${\texttt{OPT}}$ to both sides, we get
    \begin{align}
        {\texttt{OPT}}-\widehat{F}(\mb{x}^{(t+\tau)}) & \leq (1-\eta)\left( {\texttt{OPT}}-\widehat{F}(\mb{x}^{(t)})\right)+ \eta \left( r\gamma_t + \frac{\chi}{\alpha}+2(\sigma r+L_t r^{1.5})\right) + \frac{\eta^2\cdot m_F r^2}{2}
    \end{align}
    Applying the same inequality inductively gives and using $1>1-\eta$
    \begin{align}
        {\texttt{OPT}}-\widehat{F}(\mb{x}^{(t+\tau)}) &\leq (1-\eta)^{t+1}\left({\texttt{OPT}}-\widehat{F}(\mb{0}) \right)   \\
        &+\eta \Big(r\sum_{t\in I_\tau}\gamma_t + \frac{6r\sum_{t\in I_\tau}\gamma_t + 2(\sigma rT/\tau+\sum_{t\in I_\tau}L_tr^{1.5})}{\alpha \sqrt{K}}+2(\sigma rT/\tau+r^{1.5}\sum_{t\in I_\tau}L_t )\Big)+\frac{T\eta^2\cdot m_F r^2}{2}\\
        &= (1-\eta)^{t+1}{\texttt{OPT}}   \\
        &+\eta \Big(r\sum_{t\in I_\tau}\gamma_t + \frac{6r\sum_{t\in I_\tau}\gamma_t + 2(\sigma rT/\tau+\sum_{t\in I_\tau}L_tr^{1.5})}{\alpha \sqrt{K}}+2(\sigma rT/\tau+r^{1.5}\sum_{t\in I_\tau}L_t )\Big)+\frac{T\eta^2\cdot m_F r^2}{2}
        % \\
        % &\leq (1-\eta)^{t+1} {\texttt{OPT}}  + T\eta \left(r\gamma + \frac{\chi}{\alpha}+2(\sigma r+L r^{1.5})\right)+\frac{T\eta^2\cdot m_F r^2}{2}
    \end{align}
    Taking the union bound over $T/\tau$ steps and $\alpha=\sqrt{\frac{\delta\tau}{T}}$, with probability at least $1-T\cdot\alpha^2/\tau=1-\delta$, we get 
    \begin{align}
        \left( 1- (1-\eta)^{T/\tau}\right){\texttt{OPT}} 
        &\leq \widehat{F}(\mb{x}^{(T/\tau)})  \\
        &+ \eta \left(r\sum_{t\in I_\tau}\gamma_t + \frac{6r\sum_{t\in I_\tau}\gamma_t + 2(\sigma rT/\tau+\sum_{t\in I_\tau}L_tr^{1.5})}{ \sqrt{K\tau\delta/T}}+2(\sigma rT/\tau+r^{1.5}\sum_{t\in I_\tau}L_t )\right)\\
        &+\frac{T\eta^2\cdot m_F r^2}{2\tau}
        % \\
        % & = \widehat{F}(\mb{x}^{(T)}) + T\eta \left(r\gamma +O(\frac{\sqrt{T}(6r\gamma + 2(\sigma r+Lr^{1.5}) )}{\sqrt{K\tau\delta}})+2(\sigma r+L r^{1.5})\right)+\frac{T\eta^2\cdot m_F r^2}{2}
    \end{align}
    Setting $\eta=\tau/T$ yields
    \begin{align}
        \left( 1- 1/\mathrm{e}\right)\texttt{OPT} 
        &\leq \widehat{F}(\mb{x}^{(T/\tau)}) + \frac{\tau}{T} \left(r\sum_{t\in I_\tau}\gamma_t + \frac{6r\sum_{t\in I_\tau}\gamma_t + 2(\sigma rT/\tau+\sum_{t\in I_\tau}L_tr^{1.5})}{ \sqrt{K\tau\delta/T}}+2(\sigma rT/\tau+r^{1.5}\sum_{t\in I_\tau}L_t )\right)\\
        &+\frac{\tau\cdot m_F r^2}{2T}\\
        &= \widehat{F}(\mb{x}^{(T/\tau)}) + \frac{\tau rD}{T} + \frac{\sqrt{\tau}(6rD + 2r^{1.5}Q)}{ \sqrt{KT\delta}}+2\sigma r(\frac{\sqrt{\tau}}{\sqrt{KT\delta}}+1)+\frac{2\tau r^{1.5}Q }{T}+\frac{\tau\cdot m_F r^2}{2T}
    \end{align}
    \end{proof}
    \paragraph{Proof of Claim~\ref{claim:fedcg+convergence}}
    \begin{proof}[Proof of Claim~\ref{claim:fedcg+convergence}]
        % $\langle  \widetilde{\Delta}_i^{(t+\tau)}, \nabla\widehat{f}_i(\mb{x}^{(t)}) \rangle \geq \langle  \mb{v}_i^{(t)}, \nabla\widehat{f}_i(\mb{x}^{(t)}) \rangle- 2(\sigma r+L r^{1.5})/\sqrt{n}$.

        Recall the definitions of $\widetilde{\Delta}_i^{(t+\tau)}$ and $\mb{v}_i^{(t)}$. For $\widetilde{\Delta}_i^{(t+\tau)}=\frac{1}{\tau}\sum_{j=0}^{\tau-1}\widetilde{\mb{v}}_i^{(t,j)}$ where for each $j$ we have $\widetilde{\mb{v}}_i^{(t,j)}=\argmax_{\mb{y}\in \mc{P}}\langle \mb{y},\nabla\widetilde{f}_i(\mb{x}_i^{(t,j)})\rangle$. Furthermore, $\mb{v}_i^{(t)}=\argmax_{\mb{y}\in \mc{P}}\langle \mb{y},\nabla\widehat{f}_i(\mb{x}_i^{(t)})\rangle$. We have,
         \begin{align}
           \langle  \mb{v}_i^{(t)}, \nabla\widehat{f}_i(\mb{x}^{(t)}) \rangle - \langle  \widetilde{\Delta}_i^{(t+\tau)}, \nabla\widehat{f}_i(\mb{x}^{(t)}) \rangle = \frac{1}{\tau}\sum_{j=0}^{\tau-1} (\langle \mb{v}_i^{(t)},\nabla\widehat{f}_i(\mb{x}^{(t)})\rangle - \langle \widetilde{\mb{v}}_i^{(t,j)}, \nabla\widehat{f}_i(\mb{x}^{(t)})\rangle)
        \end{align}
                
        On one hand, both $\nabla\widehat{f}_i(\mb{x}_i^{(t)})$ and $\nabla\widetilde{f}_i(\mb{x}_i^{(t,j)})$ are nonnegative vectors and each $\widetilde{\mb{v}}_i^{(t,j)}$ and $\mb{v}_i^{(t)}$ corresponds to a maximum-weight independent set in the matroid, with respect to the gradient vectors, and they can be found easily by a greedy algorithm. On the other hand, equation \eqref{eq:Lipschitz-estimate} tells us $\left \| \nabla\widetilde{f_i}(\mb{x}_i^{(t,j)},\zeta_{i}^{(t,j)})  - \nabla \widehat{f_i}(\bx^{(t)}) \right \|  \leq \sigma+L_t  \sqrt{r}$.

        Let $A=\{e\mid \mb{v}_i^{(t)}(e)=1\}$ be the set of indices where $\mb{v}_i^{(t)}$ is one, similarly define $B=\{e\mid \widetilde{\mb{v}}_i^{(t,j)}(e)=1\}$. Then, by definition and equation \eqref{eq:Lipschitz-estimate} we get:
        \begin{align}
            \langle \mb{v}_i^{(t)},\nabla\widehat{f}_i(\mb{x}^{(t)})\rangle &= \sum_{a\in A} \nabla\widehat{f}_i(\mb{x}^{(t)})(a) \\
           &\geq \sum_{b\in B} \nabla\widehat{f}_i(\mb{x}^{(t)})(b) \geq \sum_{b\in B} \nabla\widetilde{f_i}(\mb{x}_i^{(t,j)},\zeta_{i}^{(t,j)})(b) - r(\sigma+L_t  \sqrt{r})
        \end{align}
        Similarly, we have
        \begin{align}
            \langle \widetilde{\mb{v}}_i^{(t,j)}, \nabla\widetilde{f_i}(\mb{x}_i^{(t,j)},\zeta_{i}^{(t,j)})\rangle 
            &= \sum_{b\in B} \nabla\widetilde{f_i}(\mb{x}_i^{(t,j)},\zeta_{i}^{(t,j)})(b) \\
            & \geq \sum_{a\in A} \nabla\widetilde{f_i}(\mb{x}_i^{(t,j)},\zeta_{i}^{(t,j)})(a)
            \geq \sum_{a\in A} \nabla\widehat{f}_i(\mb{x}^{(t)})(a)- r(\sigma+L_t  \sqrt{r}).
        \end{align}
        Putting the above two together gives us the desired bound.
    \end{proof}
    
\subsection{Gradient Estimation}

In terms of computation cost on clients' devices, we point out that the definition of multilinear extension involves summing over all subsets $S$ of $E$. There are $2^{|E|}$ such subsets, thus even computing $\widehat{g}(\mb{x})$ for a single $\mb{x}$ could take exponential time. However, we can randomly sample $m$ subsets $R_1,\dots,R_m$ of $E$ according to $\mb{x}$. Then a simple application of Chernoff's bound shows for any multilinear extension $\widehat{g}$ and $\mb{x}$, $\left| \frac{1}{m}\sum_{i=1}^m g(R_i) - \widehat{g}(\mb{x})\right|\leq \sigma\max_{S}g(S) $ \cite{Vondrak08,CalinescuCPV11}
% \begin{align}
%     \left| \frac{1}{m}\sum_{i=1}^m g(R_i) - \widehat{g}(\mb{x})\right|\leq \varepsilon\max_{S}g(S) 
% \end{align}
with probability at least $1-\mathrm{e}^{-m\sigma^2/4}$.
%Hence using $O(\frac{\log{n}}{\varepsilon^2})$ random samples, we can compute an $\varepsilon$-approximation of $\widehat{g}(\mb{x})$ with probability $1-\frac{1}{\text{poly}(n)}$.
Observe that $ \frac{\partial \widehat{g}}{\partial \mb{x}(e)}=\E[g(R\cup\{e\})]-\E[g(R)]$
where $R$ is a random subset of $E\setminus\{e\}$ sampled according to $\mb{x}$. Hence, by a similar argument, with $m$ random samples, we can compute an $\sigma$-approximation of $\nabla\widehat{g}(\mb{x})$ with $1-\mathrm{e}^{-m\sigma^2/4}$ probability. By suitably choosing $m$ as in Algorithm 2, we can assume the gradients are estimated within $(1\pm\sigma)$ accuracy with high probability.

The above discussion yields:
\begin{lemma}
\label{lem:bound-gradient-estimation}
    Let $\sigma>0$ be an error for gradient estimation and set $m=O(\log{(1/\delta)}/\sigma^2)$ for $\delta>0$. Let $\zeta_{i}^{(t,j)}=\{R_i^{(t,j,1)},\dots,R_i^{(t,j,m)}\}$ be subsets of $E$ that are sampled independently according to $i$-th client's local model $\mb{x}_i^{(t,j)}$. Let $\nabla\widetilde{f_i}(\mb{x}_i^{(t,j)},\zeta_{i}^{(t,j)})$ denote the \emph{stochastic} gradient for the $i$-th client that approximates $\nabla\widehat{f_i}(\mb{x}_i^{(t,j)})$. Then with $1-\delta$ probability we have $\|\nabla\widehat{f_i}(\mb{x}_i^{(t,j)}) -\nabla\widetilde{f_i}(\mb{x}_i^{(t,j)},\zeta_{i}^{(t,j)})\|\leq \sigma$.
\end{lemma}

%% file: Appendix-discrete-examples.tex
\section{Discrete Algorithm in Federated Setting: Examples} \label{sec:discrete}
% In this section we provide a proof for Theorem~\ref{thm:convergence-discrete}. The proof is almost identical to the one provided in \cite{Sparsification-RafieyY22}, however we present it here to make the paper self-contained. 

We consider two well-studied problems, namely \texttt{Facility Location} and \texttt{Maximum Coverage} problems and discuss the details of how to compute importance factors in federated setting efficiently and prove that the expected number of clients participating in each round is small for these two problems.

\subsection{Facility Location Problem}
Let $\mc{C}$ be a set of $N$ clients and $E$ be a set of facilities with $|E|=n$. For $c:\mc{C}\times E\to \zR$ let the $i$-th client's score function over a subset of facilities be $f_i(A)=\max_{j\in A} c(i,j)$. The objective for \texttt{Max Facility Location} is
\begin{align*}
    \max_{S\subseteq E,|S|\leq k} \left\{ F(A)=\sum_{i=1}^N \max_{j\in A} c(i,j)\right\}
\end{align*}

% Let $I$ be a set of $N$ clients and $E$ be a set of facilities with $|E|=n$. Let $c:I\times E\to \zR$ be the cost of assigning a given client to a given facility. For each client $i$ and each subset of facilities $A\subseteq E$, define $$f_i(A)=\max_{j\in A} c(i,j).$$ For any non-empty subset $A\subseteq E$, the value of $A$ is given by
% \[
%     F(A) = \sum\limits_{i\in I}f_i(A)= \sum\limits_{i\in I}\max\limits_{j\in A} c(i,j).
% \]
% For completeness, we define $F(\emptyset) = 0$. An instance of the
% \texttt{Max Facility Location} problem is specified by a tuple $(I, E, c)$. The objective is to choose a subset $A\subseteq E$ of size at most $k$ maximizing $F(A)$. 
% \begin{align}
%     \argmax_{A\subseteq E,|A|\leq k} F(A)
% \end{align}
For each client $i$  the importance factor is $\max\limits_{j\in E}\frac{c(i,j)}{ F(\{j\})}$.
% \begin{align}
% w_i = \max\limits_{A\subseteq E}\frac{f_i(A)}{F(A)}=\max\limits_{\substack{A\subseteq E}}\frac{\max\limits_{j\in A}c(i,j)}{F(A)} =  \max\limits_{j\in E}\frac{c(i,j)}{ F(\{j\})}.
% \end{align}

In many applications computing importance factors is straightforward. Let us elaborate on this with an real-world application.
\paragraph{Movie recommendation system.} Consider a movie recommendation application \cite{StanZ0K17} where each client $i$ has a user-specific utility function $f_i$ to evaluate sets of movies. The global task is to find a set of $k$ movies that are most satisfactory to \emph{all} the clients. An example is the MovieLens dataset consisting of 1 million ratings by $N=6041$ clients for $|E|=n=4000$ movies. It is in the interest of clients that we respect their privacy and they are reluctant to share their rating with a central server and other clients. We consider a well motivated objective function. Let $r_{i,j}$ denote the rating of client $i$ for movie $j$ (if such a rating does not exist set $r_{i,j}=0$). We associate to each client $i$ a facility location objective function $f_i(S)=\max_{j\in S} r(i,j)$ where $S\subseteq E$ is a subset of movies. The servers objective is $\max_{S\subseteq E,|S|\leq k} \frac{1}{N}\sum f_i(S)$. In this example, the average rating of each movie is publicly available. That is $F(\{j\})=1/N\sum_{i=1}^N f_i(\{j\})$ for each movie $j$ is publicly available. Hence, it is straightforward for each client to compute its own importance factor.

\paragraph{Computing importance factors in federated setting.} It is straightforward to see each client can compute its corresponding importance factor in a federated setting and using a secure aggregator without sharing its data with other clients. Each client $i$ sends a vector $(c(i,1),\dots,c(i,n))$ to the server and by simply summing up these vectors the server has a histogram over facilities. This histogram is then broadcasts to the clients where they can compute their own importance factor; see Algorithm~\ref{alg:importance}). Furthermore, Algorithm~\ref{alg:importance} requires only two communication rounds.

% Here we discuss that each client can compute its corresponding importance factor in a federated setting and using a secure aggregator without sharing its data with other clients. Because of the simple nature of importance factors for this problem, one can easily verify that in Algorithm~\ref{alg:importance} every client correctly computes its own importance factor. Furthermore, Algorithm~\ref{alg:importance} requires only two communication rounds.
\begin{theorem}
    In Algorithm~\ref{alg:importance}, every clients correctly computes its own importance factor. Moreover, Algorithm~\ref{alg:importance} has only two communication rounds and during each round each client requires only $O(n)$ local function evaluations.
\end{theorem}

Having $w_i$ on hand we can execute Algorithm~\ref{alg:Discrete-FL} for \texttt{Max Facility Location} problem. Note that, in this problem we are dealing with a uniform matroid of rank $k$.
\begin{theorem}
    Suppose clients' importance factors are computed using Algorithm~\ref{alg:importance} and let $\varepsilon\in (0,1)$. Algorithm~\ref{alg:Discrete-FL} after $k$ communication rounds returns a set $S$ of size $k$ such that with probability at least $1-1/n$ 
    \[ (1-1/\mathrm{e}-\varepsilon)\texttt{OPT}\leq F(S)\]
    Moreover, the expected number of clients participating during each round is $\tilde{O}(kn^2/\varepsilon^2)$.
\end{theorem}
\begin{proof}
    The approximation guarantee follows from Theorem~\ref{thm:convergence-discrete}. The expected number of clients participating in each round of Algorithm~\ref{alg:Discrete-FL} is $\tilde{O}(kn^2/\varepsilon^2)$. This is because
    \begin{align*}
        &\sum_{i=1}^n \kappa_i\leq \kappa\sum_{i=1}^n w_i\leq \tilde{O}(kn/\varepsilon^2)\sum_{i=1}^N w_i
        = \tilde{O}(kn/\varepsilon^2)\sum_{i=1}^N \max\limits_{j\in E}\frac{c(i,j)}{ F(\{j\})} \leq \tilde{O}(kn/\varepsilon^2)\sum_{j=1}^{|E|}\frac{\sum_{i\in I}c(i,j)}{F(\{j\})}\\
        &= \tilde{O}(kn/\varepsilon^2)\sum_{j=1}^{|E|}\frac{F(\{j\})}{F(\{j\})}=\tilde{O}(kn^2/\varepsilon^2)
    \end{align*}
\end{proof}

\begin{algorithm}[tb]
  \caption{Computing importance factors for \texttt{Facility Location}}
  \label{alg:importance}
  \begin{algorithmic}[1]
    \STATE {\bfseries Input: }{Ground set $E$}
    \STATE Let $\mc{O}$ be a vector of length $n$. \COMMENT{intention: $\mc{O}[i]=F(\{i\})$.}
    \FOR{each client $i$ in parallel}
        \STATE Compute $\mc{O}_i=[f_i(\{1\}),\dots,f_i(\{|E|\})]$
        \STATE Send $\mc{O}_i$ back to the secure aggregator.
    \ENDFOR
    \STATE SecAgg: $\mc{O}=\sum_{i\in [N]} \mc{O}_i$ 
    \STATE Server sends $\mc{O}$ to all clients.
    \FOR{each client $i$ in parallel}
        \STATE Compute $w_i=\max_{j\in E}\frac{c(i,j)}{\mc{O}[j]}$
    \ENDFOR
  \end{algorithmic}
\end{algorithm}

\subsection{Maximum Coverage Problem}
Let $\mc{C}=\{C_1,\dots,C_N\}$ be a set of clients and $E=\{G_1,\dots,G_n\}$ be a family of sets where each $G_i\subseteq \mc{C}$ is a group of clients. Given a positive integer $k$, in the \texttt{Max Coverage} problem the objective is to select at most $k$ groups of clients from $E$ such that the maximum number of clients are covered, i.e., the union of the selected groups has maximal size. One can formulate this problem as follows. For every $i\in [N]$ and $A\subseteq [n]$ define $f_i(A)$ as
\begin{align*}
    f_i(A)=
    \begin{cases}
      1 & \text{if there exists $a\in A$ such that $C_i\in G_a$} ,\\
      0 & \text{otherwise.}
   \end{cases}
\end{align*}
Note that $f_i$'s are monotone and submodular. Furthermore, define $F(A)=\sum_{i\in [N]}f_i(A)$ which is monotone and submodular as well. Now the \texttt{Max Coverage} problem is equivalent to 
\begin{align}
    \max_{A\subseteq [n],|A|\leq k} \left\{F(A)=\sum_{i\in [N]}f_i(A)\right\} 
    % \argmax_{A\subseteq [n], |A|\leq k} F(A).
\end{align}
For each client $C_i$, its importance factor $w_i$ is $\max\limits_{\substack{G_a\in E, C_i\in G_a}}\frac{1}{|G_a|}$.

% \begin{align*}
%     w_i = \max\limits_{\substack{A\subseteq [n], |A|\leq k}}\frac{f_i(A)}{F(A)} = \max\limits_{\substack{G_a\in E, C_i\in G_a}}\frac{f_i(\{a\})}{F(\{a\})}
%     = \max\limits_{\substack{G_a\in E, C_i\in G_a}}\frac{1}{F(\{a\})}=\max\limits_{\substack{G_a\in E, C_i\in G_a}}\frac{1}{|G_a|}.
% \end{align*}

\paragraph{Computing importance factors in federated setting.} Having a histogram over the group sizes suffices for computing the importance factors. Similar to \texttt{Facility Location} the importance factors can be computed in federated setting by having the membership histograms over the groups; see Algorithm~\ref{alg:importance-max-coverage}.
\begin{algorithm}[tb]
  \caption{Computing importance factors for \texttt{Max Coverage}}
  \label{alg:importance-max-coverage}
  \begin{algorithmic}[1]
    %\STATE {\bfseries Input: }{Ground set $E=$}
    \STATE Let $\mc{O}$ be a vector of length $n$. \COMMENT{intention: $\mc{O}[i]=|G_i|$.}
    \FOR{each client $i$ in parallel}
        \STATE Compute vector $\mc{O}_i\in \{0,1\}^n$ 
        \begin{align*}
            \begin{cases}
                \mc{O}_i[a] = 1 & \text{if } C_i\in G_a\\
                \mc{O}_i[a] = 0 & \text{otherwise}
            \end{cases}
        \end{align*}
        \STATE Send $\mc{O}_i$ back to the secure aggregator.
    \ENDFOR
    \STATE Secure aggregator computes $\mc{O}=\sum_{i\in [N]} \mc{O}_i$ and sends it to the server.
    \STATE Server sends $\mc{O}$ to all clients.
    \FOR{each client $i$ in parallel}
        \STATE Compute $w_i=\max\limits_{\substack{C_i\in G_a}}\frac{1}{|\mc{O}[a]|}$
    \ENDFOR
  \end{algorithmic}
\end{algorithm}

% Here we discuss that each client can compute its corresponding importance factor in a federated setting and using a secure aggregator without sharing which group she belongs to. Because of the simple nature of importance factors for this problem, one can easily verify that in Algorithm~\ref{alg:importance-max-coverage} every client correctly computes its own importance factor. Furthermore, Algorithm~\ref{alg:importance-max-coverage} requires only two communication rounds.
\begin{theorem}
    In Algorithm~\ref{alg:importance-max-coverage}, every clients correctly computes its own importance factor without revealing which groups they belong to. Moreover, Algorithm~\ref{alg:importance-max-coverage} has only two communication rounds.
\end{theorem}
\begin{remark}
    We point out that a simple algorithm where at each round clients share their membership with the server using SecAgg protocols is applicable here. However, such algorithm requires full client participation at each round. This can be resolved by sampling sufficiently many clients at each round. %However, the number of sampled clients as well as the approximation factor would depend on $N$.
\end{remark}
Having $w_i$, we can execute Algorithm~\ref{alg:Discrete-FL} for \texttt{Max Coverage} problem. Note that, in this problem we are dealing with a uniform matroid of rank $k$.
\begin{theorem}
    Suppose clients' importance factors are computed using Algorithm~\ref{alg:importance-max-coverage} and let $\varepsilon\in (0,1)$. Algorithm~\ref{alg:Discrete-FL} after $k$ communication rounds returns a set $S\subseteq E$ of size $k$ such that with probability at least $1-1/n$ 
    \[ (1-1/\mathrm{e}-\varepsilon)\texttt{OPT}\leq F(S)\]
    Moreover, the expected number of clients participating during each round is $\tilde{O}(kn^2/\varepsilon^2)$.
\end{theorem}
\begin{proof}
    The approximation guarantee follows from Theorem~\ref{thm:convergence-discrete}. The expected number of clients participating in each round of Algorithm~\ref{alg:Discrete-FL} is $\tilde{O}(kn^2/\varepsilon^2)$. This is because
    \begin{align*}
        &\sum_{i=1}^n \kappa_i\leq \kappa\sum_{i=1}^n w_i\leq \tilde{O}(kn/\varepsilon^2)\sum_{i=1}^N w_i
        = \tilde{O}(kn/\varepsilon^2)\sum_{i=1}^N \max\limits_{\substack{G_a\in E, C_i\in G_a}}\frac{1}{|G_a|} \leq \tilde{O}(kn/\varepsilon^2)\sum_{j=1}^{|E|}\frac{|G_j|}{|G_j|}= \tilde{O}(kn^2/\varepsilon^2) \qedhere
    \end{align*}
\end{proof}